\documentclass[9pt,journal,compsoc]{IEEEtran}

\usepackage{amssymb}
\usepackage{mathtools}
\usepackage{amsthm}
\usepackage{thm-restate}
\usepackage{soul} \newtheorem{theorem}{Theorem}[section]
\newtheorem{corollary}[theorem]{Corollary}
\newtheorem{definition}[theorem]{Definition}

\newtheorem{lemma}[theorem]{Lemma}
\newtheorem{observation}[theorem]{Observation}

\usepackage{enumerate}
\usepackage{paralist}
\usepackage{amsmath,amsfonts}
\usepackage{algorithmic}
\usepackage{algorithm}
\usepackage{array}
\usepackage[caption=false,font=normalsize,labelfont=sf,textfont=sf]{subfig}
\usepackage{textcomp}
\usepackage{stfloats}
\usepackage{url}
\usepackage{verbatim}
\usepackage{graphicx}
\usepackage{tabularx}

\allowdisplaybreaks

\usepackage{relsize}
\usepackage{blkarray}

\usepackage{color}
\usepackage{colortbl}
\usepackage{soul}
\definecolor{Gray}{gray}{.8}

\newcommand{\NN}{\mathbb{N}}
\newcommand{\BO}{\mathcal{O}}
\newcommand{\IB}{\mathbb{B}}
\newcommand{\IT}{{\mathbb{T}}}
\newcommand{\encset}[1]{\mathcal{P}_{2^{k}}(#1)}
\newcommand{\Asigj}{A_{\sigma,j}}
\newcommand{\Asigjk}{\mathcal{A}_{\sigma,j}^k}
\newcommand{\Asig}{A_{\sigma}}
\newcommand{\Asigk}{\mathcal{A}_{\sigma}^k}
\newcommand{\res}{\operatorname{res}}

\newcommand{\bigstarr}{\mathop{\raisebox{-.7pt}{\ensuremath{\mathlarger{\mathlarger{\mathlarger{*}}}}}}}
\newcommand{\mfu}{\ensuremath{\mathfrak{u}}}
\newcommand{\M}[1]{\left(\mathcal{M}_{#1}\right)}
\newcommand{\AND}{\textup{\texttt{and}}}

\newcommand{\OR}{\textup{\texttt{or}}}
\newcommand{\NOT}{\textup{\texttt{not}}}
\newcommand{\MUX}{\textup{\texttt{MUX}}}

\newcommand{\MMUX}{\operatorname{\mathcal{M}-\MUX}}

\ifCLASSOPTIONcompsoc
\usepackage[nocompress]{cite}
\else
\usepackage{cite}
\fi

\ifCLASSINFOpdf
\else
\fi

\begin{document}
\title{Small Hazard-Free Transducers}

\author{\IEEEauthorblockN{Johannes Bund\IEEEauthorrefmark{1}, \and Christoph Lenzen\IEEEauthorrefmark{2}, \and Moti Medina\IEEEauthorrefmark{3}}\\\vspace{1em}\IEEEauthorblockA{\IEEEauthorrefmark{1}\textit{Université Paris-Saclay, CNRS, ENS Paris-Saclay, Laboratoire Méthodes Formelles}\\
Gif-sur-Yvette, France \\
jobund@lmf.cnrs.fr}\\\vspace{1em}
\IEEEauthorblockA{\IEEEauthorrefmark{2}\textit{CISPA Helmholtz Center for Information Security}\\
 Saarbr\"ucken, Germany \\
 lenzen@cispa.de} \\\vspace{1em}
 \IEEEauthorblockA{\IEEEauthorrefmark{3}\textit{Alexander Kofkin Faculty of Engineering, Bar-Ilan University}\\
Ramat Gan, Israel \\
moti.medina@biu.co.il}
\thanks{This work is an extended version of the conference paper presented at ITCS 2022~\cite{bund2022small}, it is accepted for publication in IEEE Transactions on Computers.
\emph{Funding:} {This project has received funding in part from the European Research Council
(ERC) under the European Union’s Horizon 2020 research and innovation programme
(grant agreement 716562).}
{This research was supported in part by the Israel Science Foundation under Grants
867/19 and 554/23.}
{This research was funded in part
by the French Agence Nationale de la Recherche (ANR) under the project DREAMY (ANR-21-CE48-0003).}
{Johannes Bund was affiliated with the CISPA Helmholtz Center for Information Security and the Bar-Ilan University when working on this paper.}
}}

\IEEEtitleabstractindextext{\begin{abstract}
	% start input ./abstract.tex
In digital circuits, hazardous input signals are a result of spurious operation of bistable elements. For example, the problem occurs in circuits with asynchronous inputs or clock domain crossings. Marino (TC’81) showed that hazards in bistable elements are inevitable. Hazard-free circuits compute the “most stable” output possible on hazardous inputs, under the constraint that it returns the same output as the circuit on stable inputs.
Ikenmeyer et al. (JACM’19) proved an unconditional exponential separation between the hazard-free complexity and (standard) circuit complexity of explicit functions. Despite that, asymptotically optimal hazard-free sorting circuit are possible (Bund et al., TC’19). This raises the question: Which classes of functions permit efficient hazard-free circuits?
We prove that circuit implementations of transducers with small state space are such a class. A transducer is a finite state machine that transcribes, symbol by symbol, an input string of length n into an output string of length n. We present a construction that transforms any function arising from a transducer into an efficient circuit that computes the hazard-free extension of the function. For transducers with constant state space, the circuit has asymptotically optimal size, with small constants if the state space is small.

     % end input ./abstract.tex
 \end{abstract}

}

\maketitle

\IEEEdisplaynontitleabstractindextext

\IEEEpeerreviewmaketitle

% start input ./intro.tex
\IEEEraisesectionheading{\section{Introduction}\label{sec:intro}}

\IEEEPARstart{D}{igital} circuit design relies on a fundamental abstraction of the physical
world. Electric voltages transmitted by wires are mapped to Boolean values,
where high voltages correspond to
$1$ ($\mathrm{true}$) and low voltages to $0$ ($\mathrm{false}$).
By virtue of this abstraction, the behavior of digital circuitry can be
described by Boolean formulas.
However, this description does not account for the behavior of digital circuits
in all cases:
it offers no way of representing signals that are unstable, transitioning,
oscillating, glitching, etc.

In this work, we study a classic extension of Boolean logic due to
Kleene~\cite[\S64]{kleene52meta}, which allows for the presence of unspecified
signals.
In the following, we refer to the Boolean values $\IB\coloneqq\{0,1\}$ as
\emph{stable,} while the additional third logical value $\mfu$ is the
\emph{unstable} value.
The resulting ternary set of logic values is denoted by
$\IT\coloneqq\{0,1,\mfu\}$.
Intuitively, $\mfu$ may evaluate to any stable state at any
point in the circuit, regardless of previous evaluations.
We regard $\mfu$ as a ``superposition'' of $0$ and $1$.

In Kleene logic, the basic gates\footnote{The specific choice of basic gates
does not matter, see~\cite{ikenmeyer18complexity}; hence, we stick to
$\AND$, $\OR$, and $\NOT$.} output a stable value if and only if the stable
inputs already determine this output. The natural extension of the basic gates
$\AND$, $\OR$, and $\NOT$ is given in Table~\ref{tab:gates}.
Given a circuit $C$; the function $C\colon \IT^n\to \IT^m$ is formally defined 
by induction over the circuit structure.

A key difference between Kleene logic and Boolean logic can be seen when regarding the laws
    of excluded middle and non-contradiction.
In Boolean logic both laws are true, but in Kleene logic they do not hold:
\[\OR(x,\NOT(x))\neq 1 \quad\text{ and }\quad\AND(x,\NOT(x))\neq 0,\] as $\NOT(\mfu)=\mfu$ and
$\AND(\mfu,\mfu)=\OR(\mfu,\mfu)=\mfu$. This drastically distinguishes $\mfu$
from an unknown Boolean value.
Accounting for this limitation and the fact that constant stable inputs can
easily be provided, we allow for constant-$0$ (and constant-$1$) gates in
addition to the basic gates $\OR$, $\AND$, and $\NOT$.

In general, digital logic cannot detect or prevent the propagation of unstable
signals~\cite{Mar81}, matching the above limitation present in Kleene logic.
Furthermore, \cite{friedrichs18containing}
extends the above model for combinational circuits to general clocked circuits.
Thus, understanding combinational circuits in the model we use in this paper is
equivalent to understanding the worst-case propagation of unstable signals in
general-purpose digital logic in a precise sense.

It has been shown that CMOS logic gates implement the stated
specification~\cite{BLM-TC-19}, i.e., standard cells meet the specification
given in Table~\ref{tab:gates}.
Hence, the presented circuits can be implemented in standard processes. 
This allows for easy integration in state-of-the-art computer architectures 
and embedded systems.

\subsection{Hazards and Hazard-free Circuits}

\begin{table}
  \caption{Behaviour of basic gates $\AND$, $\OR$, and $\NOT$.}
  \label{tab:gates}
\centering
\begin{tabular}{c|ccc}
  \AND & 0 & 1 & \mfu\\ \hline
  0 & 0 & 0 & 0\\
  1 & 0 & 1 & \mfu\\
  \mfu & 0 & \mfu & \mfu
\end{tabular}
\qquad
\begin{tabular}{c|ccc}
  \OR & 0 & 1 & \mfu\\ \hline
  0 & 0 & 1 & \mfu\\
  1 & 1 & 1 & 1\\
  \mfu & \mfu & 1 & \mfu
\end{tabular}
\qquad
\begin{tabular}{c|c}
\NOT &\\ \hline
0 & 1\\
1 & 0\\
\mfu & \mfu
\end{tabular}
\end{table}

We strive for circuits that behave similarly to basic gates when receiving
unstable inputs. That is, given all inputs, if changing an input bit from $0$ to $1$ does not affect the output, then setting this input bit to $\mfu$
should also not affect the output. 
A Hazard-free circuit should tolerate a fault at that specific input.
To formalize this concept, we make use of two
operations. The first is the
\emph{superposition,} which results in the unstable value $\mfu$ whenever its
inputs do not agree on a stable input value.

\begin{definition}[Superposition]\label{def:superposition}
Denote the \emph{superposition} of two bits by the operator $*\colon\IT \times \IT\rightarrow\IT$.
For $x,y\in \IT$, $x*y=x$ if $x=y$ and $x*y=\mfu$ otherwise. We extend the
$*$-operation to $x,y\in \IT^n$ by applying it at each bit $i\in\{1,\ldots,n\}$, such that
\[ (x*y)_i = \begin{cases} x_i &\text{if } x_i=y_i,\\\mfu &\text{otherwise.} \end{cases}\]
\end{definition}
The $*$-operation is
associative and commutative, hence for $X\subseteq \IT^n$ we define $\bigstarr
X \coloneqq x_1*\hdots *x_n$, where $x_1,\ldots,x_n$ is an arbitrary
enumeration of the elements of $X$.
The second operation, called \emph{resolution}, maps from ternary strings to
sets of Boolean strings by replacing each $\mfu$ with both stable values.
\begin{definition}[Resolution]\label{def:resolution}
Denote the \emph{resolution} by $\res\colon\IT\rightarrow\mathcal{P}(\IB)$.
For $x\in\IT$, $\res(x)=\{0,1\}$ if $x=\mfu$ and $\res(x)=\{x\}$ otherwise.
We extend this to bit strings of length $n$ in a natural way, by setting it for
$x\in \IT^n$
\begin{equation*}
\res(x)\coloneqq \{y\in
\IB^n\,|\,\forall i\in \{1,\ldots,n\}\colon x_i\neq \mfu \Rightarrow
y_i=x_i\}.
\end{equation*}
\end{definition}

For notational convenience, we extend all functions $f\colon X\rightarrow Y$ to
sets of inputs $X'\subseteq X$, by applying them to each element of the input,
i.e., $f(X')\coloneqq\{f(x)\,|\,x\in X'\}$. For instance,
$\res(\{0\mfu0,1\mfu0\})=\{000,010,100,110\}$. First, we observe that superposition and
resolution are not inverse functions.

\begin{observation}\label{obs:resstar}
  Let $X\subseteq\IT^n$, then taking the resolution of the superposition
  of all strings in $X$ may add further strings, i.e., $X\subseteq\res(\bigstarr X)$.
\end{observation}

Note that a strict superset relation is possible. For instance, if there are two strings $x,y\in X$
that disagree on more than $\log(|X|)$ positions,
then $|\res(\bigstarr X)|>2^{\log(|X|)}=|X|$ and hence $X\subset\res(\bigstarr X)$.
For example, if $X=\{101,110\}$, then $\res(\bigstarr X)=\res(1\mfu\mfu)=\{100,101,110,111\}$.

Recall that $\AND$ guarantees a stable output of $0$ if at least one of
its inputs is $0$ -- even if the other input is $\mfu$. As a toy example,
consider two circuits implementing a conjunction of $x$ and $y$. First
(needlessly involved)
$\AND(\OR(x,\NOT(\OR(y,\NOT(y)))),y)$ and second $\AND(x,y)$.
Under Boolean inputs, these expressions are equivalent, and so are the circuits.
In contrast, for inputs $x=0$ and $y=\mfu$, we get that
\[\AND(\OR(0,\NOT(\OR(\mfu,\NOT(\mfu)))),\mfu)=\mfu\neq 0=\AND(0,\mfu).\]
The first implementation has an
unstable output, even though the stable input of $0$ already determines that the
output should be $0$. This is referred to as a \emph{hazard.}

A convenient formalization of this concept is as follows. For a circuit $C$
with $n$ inputs and $m$ outputs, denote by $C(x)\in \IT^m$ the output it
computes on input $x\in \IT^n$. We say that $C$ \emph{implements} the Boolean
function $f\colon \IB^n\to \IB^m$, iff $C(x)=f(x)$ for all $x\in \IB^n$. The
desired behavior of the circuit is given by the \emph{hazard-free extension} of
$f$.
\begin{definition}[Hazard-free Extensions]\label{def:hfree}
For function $f\colon \IB^n\to \IB^m$, denote by $f_{\mfu}\colon \IT^n\to \IT^m$
its \emph{hazard-free extension}, which is defined by $f_{\mfu}(x)\coloneqq
\bigstarr_{y\in \res(x)}f(y)$.
\end{definition}
For any Boolean function $f$, there is a circuit implementing the hazard-free
extension of $f$~\cite{huffman57design}.
The hazard-free extension is the ``most precise'' extension of $f$ computable by
combinational logic. It is easy to show that
$(f_{\mfu})_i(x)=\mfu$ entails that $C_i(x)=\mfu$ for any circuit $C$
implementing $f$: Restricted to Boolean inputs $C$ and $f_{\mfu}$ are
identical. Changing input bits to $\mfu$ can change the output bits to $\mfu$ only,
and if $(f_{\mfu})(x)$ is $\mfu$ at position $i$ then $C(x)$ is $\mfu$ at position $i$.
In contrast, $C$ has a \emph{hazard} at $x\in \IT^n$ iff it deviates
from the desired behavior.
\begin{definition}[k-Bit Hazards]
Circuit $C$ implementing $f\colon \IB^n\to \IB^n$ has a \emph{hazard} at $x\in
\IT^n$ iff $C(x)\neq f_{\mfu}(x)$.
If, for some $k\in\NN$, $x$ contains at most $k$ many $\mfu$'s and $C$ has a hazard on
$x$, it is called a \emph{k-bit} hazard. $C$ is $k$-bit hazard-free if it
has no $k$-bit hazards, i.e., if
for every $x\in\IT^n$ where $\mfu$ appears at most $k$ times, $C(x) = f_{\mfu}(x)$.
\end{definition}
Note that $C$ having a hazard at $x$ is equivalent to $C(x)$ containing more
$\mfu$'s than necessary.
The smallest non-trivial example of a hazardous circuit is given by a naive
implementation of a multiplexer circuit. A multiplexer has the Boolean
specification $\MUX(a,b,s)=a$ if $s=0$ and $\MUX(a,b,s)=b$ if $s=1$. It can be
implemented by the circuit corresponding to the Boolean formula
$\OR(\AND(a,\NOT(s)),\AND(b,s))$, resulting in a hazard at
$(1,1,\mfu)$, cf.~\cite{ikenmeyer18complexity}.

Ikenmeyer et al.\ proved unconditional lower bounds on the complexity of
hazard-free circuits implementing explicit
functions~\cite{ikenmeyer18complexity}. More precisely, they show exponential
gaps between the size of several Boolean circuits and their hazard-free
counterparts. Furthermore, they show that hazard-free verification circuits for NP-hard
problems cannot be of polynomial size unless the circuit equivalent of
P\,$=$\,NP holds. On the other hand, there are efficient implementations
of sorting networks that avoid certain hazards~\cite{BLM-TC-19}, showing that
hazard-free implementations do not always come at a high cost. This leads to the following
question:
\begin{quote}
\emph{Which classes of Boolean functions allow for an efficient hazard-free
implementation?}
\end{quote}

\noindent
The key ingredient to the result from~\cite{BLM-TC-19} is a circuit
implementation of a finite-state transducer parsing the inputs, which leads to
a parallel prefix computation (PPC) task. Ladner and
Fischer~\cite{ladner1980parallel} presented a general framework providing an
efficient circuit implementation of arbitrary (small) transducers, giving rise
to the most efficient adder circuits known to date.
While the Ladner and Fischer framework fails to yield hazard-free circuits, the
result in~\cite{BLM-TC-19} suggests the possibility of a general hazard-free
construction. Despite being a somewhat specialized class of circuits, the fact
that addition can be phrased as a PPC problem renders hazard-free transducer
circuits a key stepping stone towards hazard-free arithmetics.

\subsection{Transducers}
A deterministic finite-state transducer is a finite-state machine that outputs
a symbol on each state transition. We phrase our results for \emph{Mealy
machines}~\cite{mealy1955method}, but our techniques are not specific to this
type of transducer.

\begin{definition}[Mealy Machine]
A Mealy machine $T=(S,s_0,\Sigma,\Lambda,t,o)$ is a $6$-tuple, where
\begin{enumerate}[(i)]
  \item $S$ is the finite set of states,
  \item $s_0\in S$ is the starting state,
  \item $\Sigma$ is the finite input alphabet,
  \item $\Lambda$ is the finite output alphabet,
  \item $t\colon S\times \Sigma \to S$ is the state transition function, and
  \item $o\colon S\times \Sigma \to \Lambda$ is the output function.
\end{enumerate}
\end{definition}

Each Mealy machine induces a \emph{transcription function} mapping a string of
input symbols to a string of output symbols of the same length.

\begin{definition}[Transcription Function $\tau$]
For $n\in \NN$ and Mealy machine $T=(S,s_0,\Sigma,\Lambda,t,o)$, the
\emph{transcription function $\tau_{T,n}\colon \Sigma^n \to \Lambda^n$}
is given in the following way. Define for
$i\in \{1,\ldots,n\}$ and $x \in \Sigma^n$ the \emph{state $s_i$ after
$i$ steps} inductively via $s_i\coloneqq t(s_{i-1},x_i)$. Then
$\tau_{T,n}(x)_i\coloneqq o(s_{i-1},x_i)$.
\end{definition}

Note that every Boolean function $f\colon \IB^n\rightarrow\IB^m$ can
(essentially) be realized by a deterministic finite-state transducer. A simple
implementation could read the entire input string $x$ and output $f(x)$ on
the reception of the last input symbol. This approach, however, requires an
exponential number of states $|S|\in\BO(2^n)$ to memorize the input.
Accordingly, it is of interest to consider \emph{small} transducers. In
particular, important basic operations, like addition, $\max$, and $\min$, can
be implemented by constant-size transducers.

\subsection{Our Contribution}

In this work, we establish that constant-size transducers allow for an efficient
hazard-free circuit implementation.
Denoting by $\ell$ and $m$ the (constant) number of bits encoding an input
symbol and an output symbol, respectively, by $n$
the length of the input string, by $|S|$ the number of states of the
transducer, and by $k$ an upper bound on the number of metastable bits in the
input, our main result is as follows.

\begin{restatable*}{theorem}{mainres}\label{thm:trans}
For any integers $k\in\NN$, $\ell,m,n\in \NN_{>0}$ (with $k\leq n$) and Mealy machine
$T=(S,s_0,\Sigma=\IB^{\ell},\Lambda\subseteq\IB^m,t,o)$, there is a $k$-bit hazard-free
circuit implementing $\tau_{T,n}$. For $\kappa \coloneqq \sum^{\min\{|S|, 2^k\}}_{i=0}\binom{|S|}{i}$
and $\lambda\coloneqq\min\{m, 2^{|S|\cdot|\Sigma|}\}$
the circuit has
\begin{align*}
\text{size}\quad &
    \BO\left((\kappa^3 + (2^{\ell}/\ell)\kappa^2 + 2^{\ell}\kappa\lambda)n\right)\\*
    \text{and depth}\quad &
    \BO\left(\log\kappa\log n + \ell\right)
    \,.
  \end{align*}
\end{restatable*}

We remark that the proof of Theorem~\ref{thm:trans}
shows that we can save a factor of $\kappa$ in the third term, provided that
the preimage of $1$ under $o(\cdot,\sigma)_j$ (i.e., bit $j$ of the output function with the
second input fixed to $\sigma$) has size at most $2^k$ for each $\sigma \in
\Sigma$ and $j\in[m]$.
In this case, there is a $k$-bit hazard-free circuit implementing $\tau_{T,n}$
of size $\BO\left((\kappa^3 + 2^{\ell}\kappa^2 + 2^{\ell}\lambda)n\right)$.

The asymptotic complexity depends on $k$, the upper bound on the numbers of $\mfu$'s.
For $k\in\NN$
we consider two cases: $2^k\geq |S|$ and $2^k<|S|$.
Let in both cases $\lambda\coloneqq\min\{m, 2^{|S|\cdot|\Sigma|}\}$.
If $2^k\geq |S|$, we apply the trivial bound
of $2^{|S|}$ for the sum over the binomial coefficients $\kappa$.
Note that in this case, trivially the preimage of $1$ under $o(\cdot,\sigma)$
has size at most $2^k$.
Hence, as discussed above, the factor of $\kappa$ in the third term of the
size bound can be removed.
This gives us the following size and depth bounds for a fully hazard-free
implementation.

\begin{corollary}\label{cor:general}
For any integers $\ell,n\in \NN$ and Mealy machine
$T=(S,s_0,\Sigma=\IB^{\ell},\Lambda,t,o)$, the transcription function $\tau_{T,n}$
can be implemented by a hazard-free circuit
\begin{align*}
  \text{of size}\quad &\BO((2^{3|S|}+2^{2|S|+\ell}/\ell+2^\ell\lambda)n)\\
  \text{and depth}\quad &\BO(|S|\log n + \ell)\,.
\end{align*}
\end{corollary}

We stress that this result stands out against the lower bound
from~\cite{ikenmeyer18complexity}, which proves an exponential dependence of the
circuit size on $n$, for any general construction of hazard-free circuits.
While the above theorem incurs exponential overheads in terms of the size
of the transducer, the dependence on $n$ is asymptotically optimal. Thus, for
constant-size transducers, we obtain asymptotically optimal hazard-free
implementations of their transcription functions, both concerning size and
depth. More generally, Theorem~\ref{thm:trans} shows that the task of
implementing transcription functions is fixed-parameter tractable concerning
$\max\{\ell,|S|\}$.

If $2^k<|S|$, the Binomial Theorem~\cite{graham1994concrete}
provides a stronger bound for $\kappa$, the sum over the binomial coefficients.
Note that the respective factor in the third term of the size bound can still be
removed if the output function satisfies the above requirement, but this does
not hold in general.
\begin{corollary}\label{cor:smallk}
Given integers $k,\ell,n\in \NN$ and Mealy machine
$T=(S,s_0,\Sigma=\IB^{\ell},\Lambda,t,o)$, such that $2^k<|S|$, the transcription function $\tau_{T,n}$
can be implemented by a $k$-bit hazard-free circuit
\begin{align*}
  \text{of size}\quad &\BO((|S|^{3\cdot2^k}+(2^\ell/\ell)|S|^{2\cdot2^{k}}+2^\ell|S|^{2^k}\lambda)n)\\
  \text{and depth}\quad &\BO(2^k\log(|S|)\log n +\ell)\,.
\end{align*}
\end{corollary}

The main insight underlying the proof of Theorem~\ref{thm:trans} is an
understanding of how the encoding of a piece of information (such as an input) affects the ability of the circuit
to keep track of this information.
Due to the ambiguity presented by
$\mfu$ signals, naive encodings may lose information crucial for determining a
stable output, which cannot be recovered later. We tackle this problem by introducing a
``universal'' encoding that explicitly stores for each $A\subseteq S$ (of size at most $2^k$) whether the
state machine is currently in \emph{some} state from $A$. This redundancy
is sufficient to eliminate $k$-bit hazards, yet is affordable when $|S|$ or $k$ are small.

\subsubsection*{Application}
The presented framework is primarily of theoretical interest, as it shows good asymptotic behavior for hazard-free circuits arising from transducers.
Certainly, tailored solutions offer better constants for specific transducer functions.
However, the construction presented in this work is general in the sense that any given input transducer
will yield a hazard-free circuit. 
Hence, the construction can be used to facilitate the synthesis of hazard-free circuits.
Synthesis is the engineering process where the input is a finite state transducer and the output is a combinational circuit.
The construction presented in this paper can be used as a fully automatic synthesis tool
that generates hazard-free circuits by design.

\subsubsection*{Organization of this article}
We discuss related work in Section~\ref{sec:related}. With the help of a toy example,
Section~\ref{sec:example} builds intuition and presents the key ideas needed to obtain
Theorem~\ref{thm:trans}.
That is, Section~\ref{subsec:ppc} explains why the Ladner and
Fischer framework fails, and Section~\ref{subsec:encoding} introduces the encoding
used to resolve the main shortcoming of their approach. We prove Theorem~\ref{thm:trans}
in Section~\ref{subsec:poutline}.
In later section we discuss at hand of detailed examples why the choice of input encoding is important
(Section~\ref{sec:inp_enc}) and the effects of bounding $k$ (Section~\ref{sec:example_sec}).
 % end input ./intro.tex
 % start input ./related.tex
\section{Related Work}\label{sec:related}
\subsubsection*{Applications of Hazard-Free Circuits}
If timing constraints for accessing bistable elements, such as flip-flops or
latches, are violated, they may become metastable. That is, their output signal
exhibits an intermediate voltage between high ($1$) and low ($0$) for an unknown
amount of time before it resolves to either one of them. Downstream circuit
components may respond as if subjected to an input of $1$ or $0$, or produce an
intermediate output voltage themselves, where the responses of different
components may be in conflict.

By the late '70s, there was an intense
debate among hardware developers whether the problem of metastability can
be dealt with deterministically, by suitable design of
circuits~\cite{marino1977effect,pechoucek1976anomalous,stucki79,wormald77synchronizer}.
Marino~\cite{Mar81} proved by a topological argument that no circuit (with
non-constant output) can avoid, resolve, or detect metastability in all cases.
Ever since, the standard approach to evading metastability in applications
where timing constraints cannot be guaranteed have been synchronizers,
see, e.g.,~\cite[Chap.~2]{kinniment08synchronization}. Synchronizers
trade time for the decreased probability of ongoing metastability (and thus
resulting in errors).

As mentioned earlier, modeling propagation of metastability in a worst-case
fashion matches Stephen Cole Kleene's ``strong logic of
indeterminacy''~\cite[\S64]{kleene52meta}. Characterization of the complexity of
hazard-free circuits, thus, is of immediate relevance for avoiding the use of
synchronizers, eliminating both incurred delays and the (remaining) probability
of error due to deterministic guarantees.
A non-trivial example of this was given in \cite{BLM-TC-19}, where Gray code
inputs that may have some metastability are sorted deterministically, with
only constant-factor overheads compared to optimal sorting networks in Boolean
logic.

To our knowledge, the first hazard-free multiplexer was published by
Goto~\cite{goto49relay}, but remained unnoticed by the Western world for
decades. Huffman~\cite{huffman57design} provided the first
general construction of metastability-containing circuits.
Goto and Huffman make no mention of Kleene's logic, developing their terms.
Similarly, related work on
cybersecurity~\cite{hu12complexity,tiwari09flow} appears to come
up with the concept independently again.
See~\cite{brzozowski01algebras} for a survey covering some of these articles and
discussing different logics.
Further applications where discussed in \cite{ikenmeyer18complexity}.
In our opinion, all of this goes to show that the questions we study in this
paper are fundamental and of widespread interest.

\subsubsection*{Complexity of Hazard-free Circuits}
In~\cite{ikenmeyer18complexity} it was shown that for monotone functions, their
\emph{hazard-free complexity} (i.e., the size of the smallest hazard-free
implementation) equals their monotone complexity (i.e., the size of the smallest
implementation without negation gates). This yields several unconditional
lower bounds as corollaries of results on monotone circuits.
In particular, an exponential separation between hazard-free and standard
circuit complexity follows from~\cite{alon87monotone,tardos88gap}, and the naive
monotone circuit of cubic size for Boolean matrix multiplication is
optimal~\cite{mehlhorn76monotone,paterson75complexity}. These lower bounds where
complemented by a general construction yielding circuits of size $n^{\BO(k)}|C|$
without $k$-bit hazards, where $C$ is an arbitrary circuit implementing the desired
function. Thus, for constant $k$, the overhead for removing $k$-bit hazards is
polynomial in $n$. The above separation result implies that an overhead of
$2^{k^{\Omega(1)}}$ is necessary, but it remains open whether the task is
fixed-parameter tractable w.r.t.\ $k$.

Jukna~\cite{jukna2021notes} strengthened the results from
Ikenmeyer et al.~\cite{ikenmeyer18complexity} on the gap between unconstrained
and hazard-free circuit complexity by proving that every minimal 
hazard-free circuit for a monotone Boolean function
must be monotone.
Moreover, Jukna showed that in general any Boolean function
$f\colon\IB^n\rightarrow\IB$ can be implemented by a hazard-free circuit of size
$\BO(2^n/n)$. We remark that applying this to the transcription function
$\tau_{T,n}\colon\IB^{\ell n}\rightarrow\IB^{m n}$ results in a circuit of size
$\BO((2^{\ell n}/(\ell n))m n)=\BO((2^{\ell n}/\ell)m)$, which is much larger than the circuit presented in this
work.

In contrast, some functions and specific hazards admit much more efficient
solutions, e.g.\ the optimal sorting networks in \cite{BLM-TC-19}.
If the possible positions of $\mfu$ inputs are restricted to the index
set $I$, a construction based on hazard-free multiplexers avoids the respective
hazards with a circuit of size
$\BO(2^{|I|}|C|)$~\cite[Lemma~5.2]{ikenmeyer18complexity}, where $C$ is as
above. This can be seen as combining speculative
computing~\cite{tarawneh12hiding,tarawneh14eliminating} with hazard-free multiplexers.

Furthermore, we remark that the lower bound can be circumvented by using non-standard
non-combinational logic~\cite{friedrichs18containing}. Using clocked circuits
and so-called masking registers, $k$-bit hazards can be eliminated with factor
$\BO(k)$-overhead. Masking registers also strictly increase the computational
power of the system with each clock cycle. However, in this paper, we consider
combinational logic only.

\subsubsection*{Transducers}
Our approach can be seen as an extension of the work of Ladner and
Fischer~\cite{ladner1980parallel}. This celebrated result yields the only
asymptotically optimal adder constructions known to date,
cf.~\cite{swartzlander15arithmetic}. Alongside the result for binary addition,
the authors pointed out the general applicability of their parallel prefix computation
(PPC) framework: for any transcription function, it allows constructing a
circuit implementing it. However, as we discuss in detail in
Section~\ref{sec:example}, their approach cannot be applied to our setting, as
it does not take into account the uncertainty imposed by unstable inputs.

Our approach might also remind the reader of the power set
construction~\cite[Thm.~1.39]{sipser2012introduction}, which translates a
non-deterministic finite-state automation into a deterministic one operating on
the power set of the state space. This analogy is correct to the extent that we
seek to maintain information on the set of states that are reachable by
resolutions of the input. However, Kleene logic has the fundamentally different
characteristic that the choice of encoding (e.g.\ of states) affects to what extent the circuit
can keep track of the encoded information. In a nutshell, we prove that it is
sufficient to maintain a bit vector indicating for each element $A\subseteq
S$ of the power set whether, given the input, all states that could have
been reached by the state machine are a subset of $A$. This resolves an issue
that has no connection to the original power set construction.
 % end input ./related.tex
 % start input ./example.tex
\section{Extending the PPC Framework to Hazard-free Circuits}\label{sec:example}

In this section, we walk the reader through the main ideas underlying our
framework, at the hand of a simple running example, and then prove our main result.
In Section~\ref{subsec:ppc} we demonstrate that a naive application of the
parallel prefix framework~\cite{ladner1980parallel} results in circuits that are
not hazard-free. We then use the running example to illustrate how to overcome
this hurdle and to obtain a hazard-free circuit by making use of the universal
encoding introduced in Section~\ref{subsec:encoding}. Finally, we prove
Theorem~\ref{thm:trans} in Section~\ref{subsec:poutline}.

The running example we use throughout this section is an extremely simple
transducer: it simply shifts the input sequence by one bit, outputting a $0$ on
the reception of the first symbol; see Figure~\ref{fig:ex_trans} for an illustration.
It has two states (referred to as
$0$ and $1$), which are used to keep track of the most recently processed input
bit. Hence, the transition and output functions are obvious: the automaton
transitions to state $s\in\{0,1\}$ on the reception of input $s$, and outputs $s$
when leaving the state $s$. Thus, the transducer is formally specified by the
6-tuple ($S\coloneqq\{0,1\}$, $s_0\coloneqq 0$, $\Sigma\coloneqq\{0,1\}$,
$\Lambda\coloneqq\{0,1\}$, $t(s,i)=i$, $o(s,i)=s$).

This transducer is a toy example, and it is pointless to construct a
circuit implementing its transcription function -- this is easily achieved by
suitable rewiring the inputs instead. However, the shift transducer serves as a
minimal example for illustrating both the obstacle we need to overcome and the
general solution we provide for doing so.

\begin{figure}
  \centering %\scriptsize
  \includegraphics{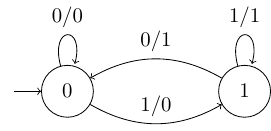}
  \caption{The shift transducer delays the input by one symbol. It serves as a
  running example.}
  \label{fig:ex_trans}
\end{figure}

\subsection{The Classic PPC Framework}\label{subsec:ppc}

In their work, Ladner and Fischer observe that any transcription function $\tau_{T,n}$ on
inputs $x\in \IB^n$ can be efficiently implemented as a circuit by following
four steps. Given an encoding of functions $S\to S$,
each step takes the output of the previous step as input.
For each $i\in\{1,\hdots, n\}$:\vspace{.2em}
\begin{tabularx}{\columnwidth}{cX}
  (Step~1) & compute the encoding of $t_{x_i}$,
    where $t_\sigma$ is the restricted transition function for symbol $\sigma\in\Sigma$, where
    $t_\sigma\coloneqq t(\cdot,\sigma)\colon S\rightarrow S$,\\
  (Step~2) & compute the composition
    $\pi_{i}\coloneqq t_{x_{i}} \circ \hdots \circ t_{x_1}$
    of restricted transition functions,\\
  (Step~3) & compute the $i$-th state $s_i=\pi_i(s_0)$, i.e., evaluate $\pi_i$ at $s_0$, and\\
  (Step~4) & compute the $i$-th output $o(s_{i-1}, x_i)=\tau_{T,n}(x)_i$.
\end{tabularx}\smallskip

\noindent
Steps $1$, $3$, and $4$ can be performed independently and hence in parallel for
each $i$, based on the output of previous steps. This means that each of them can
be performed by $n$ copies of a circuit whose size (and thus depth) depends only
on the transducer. In contrast, Step~$2$, the computation of all prefixes,
inherently relies on information across all $i$'s. To achieve small depth
without blowing up the circuit size, Ladner and Fischer exploit the
associativity of function composition.

For a constant-size Mealy machine, Steps $1$, $3$, and $4$ can be performed by
circuits of size $\BO(n)$ and depth $\BO(1)$, and Step~$2$ can be done by a
circuit of size $\BO(n)$ and depth $\BO(\log n)$. In their argument showing
this, Ladner and Fischer encode the space of functions $S\to S$ as Boolean
$|S|\times|S|$ matrices.\footnote{We remark that it would be more efficient to
encode these functions by listing their values, reducing the size of the
encoding from $|S|^2$ to $|S|\lceil\log |S|\rceil$. However, for $|S|\in \BO(1)$
this does not affect the asymptotics.} Functional composition (Step~$2$), hence,
becomes Boolean matrix multiplication, while evaluation of functions (Step~$3$)
becomes Boolean matrix-vector multiplication.

Applying this to our example, states $0$ and $1$ of our transducer are
represented by the column unit vector $e^{(0)}=\binom{1}{0}$ and $e^{(1)}=\binom{0}{1}$,
respectively. Representing $t_{\sigma}\colon S\to S$ as a Boolean matrix in a natural way,
the column corresponding to state $s$ is the unit vector $e^{(t(s,\sigma))}$. Assuming that we make the best
effort and use a hazard-free circuit for computing the encodings of $t_{\sigma}$,
the hazard-free extension determines what our circuit will compute when
receiving $\mfu$ as an input symbol. Denote by $\mathcal{M}_{t_\sigma}$ the matrix
computed by this hazard-free circuit for (the encoding of) the transition function restricted to the input symbol
$\sigma$. We thus obtain
\begin{align*}
  &\mathcal{M}_{t_0}=
  \begin{pmatrix}
    1 & 1 \\
    0 & 0
  \end{pmatrix},\quad
  \mathcal{M}_{t_1}=
  \begin{pmatrix}
    0 & 0 \\
    1 & 1
  \end{pmatrix},\\
  &\mathcal{M}_{t_{\mfu}}=
  \begin{pmatrix}
    1 & 1 \\
    0 & 0
  \end{pmatrix}\bigstarr
  \begin{pmatrix}
    0 & 0 \\
    1 & 1
  \end{pmatrix}=
  \begin{pmatrix}
    \mfu & \mfu \\
    \mfu & \mfu
  \end{pmatrix},
\end{align*}
where the $*$ operator is applied component-wise: as each entry of the computed
matrix depends on whether the input symbol was $0$ or $1$, $\mfu$ must result
in the all-$\mfu$ matrix.

Function composition corresponds to Boolean matrix multiplication, i.e., for
restricted functions $t_\sigma$ and $t_{\sigma'}$ ($\sigma,\sigma'\in\Sigma$),
$\mathcal{M}_{t_\sigma \circ t_{\sigma'}} = \mathcal{M}_{t_\sigma} \cdot \mathcal{M}_{t_{\sigma'}}$, where $\cdot$ denotes the
Boolean matrix multiplication operator. Similarly, function evaluation
corresponds to matrix-vector multiplication, meaning that the framework
stipulates to compute the encoding of $\pi_i$ as $\mathcal{M}_{t_{x_{i}}}\cdot
\ldots \cdot \mathcal{M}_{t_{x_1}}$ and hence $s_i$ as $\mathcal{M}_{t_{x_{i}}}\cdot
\ldots \cdot \mathcal{M}_{t_{x_1}}\cdot e^{(s_0)}$.
Again, we make the best effort, i.e., assume that hazard-free circuits are used.
Therefore, the circuit will compute $\mathcal{M}_{t_{x_{i-1}}}\cdot_{\mfu} \ldots
\cdot_{\mfu} \mathcal{M}_{t_{x_1}}\cdot_{\mfu} e^{(s_0)}$.

Finally, the $i$-th output bit is computed according to the output function
by mapping $e^{(0)}$ to output $0$ and
$e^{(1)}$ to output $1$. Note that the redundant representation allows for some freedom:
we can choose for $\binom{0}{0}$ and $\binom{1}{1}$ whether to map them to $0$
or $1$, respectively. As these state vectors cannot occur anyway, this choice
has no impact on stable inputs. It might, however, affect the behavior of a
(hazard-free) circuit confronted with unstable inputs.

{\renewcommand*{\arraystretch}{1.2}
\begin{table*}\caption{Application of the Ladner and Fischer approach to the example transducer
  for two different input strings. Top: stable input word $0010$, Bottom: unstable
  input word $0\mfu10$. Gray area: these values do not match the results from a
  hazard-free computation.}
\label{tab:examplerun}
\centering
\begin{tabular}{|c|c|c c c c c|}
    \hline
     & $i$ & 0 & 1 & 2 & 3 & 4 \\ \hline
    input & $x_i$ & - & $0$ & $0$ & $1$ & $0$ \\ \hline
    Step~1, function encoding & $\mathcal{M}_{t_{x_i}}$ & - & $\mathcal{M}_{t_0}$ & $\mathcal{M}_{t_0}$ &
    $\mathcal{M}_{t_1}$ & $\mathcal{M}_{t_0}$\\ \hline
    Step~2, function composition & $\mathcal{M}_{\pi_{i}}$ & - & $\mathcal{M}_{t_0}$ & $\mathcal{M}_{t_0}$ &
    $\mathcal{M}_{t_1}$ & $\mathcal{M}_{t_0}$\\ \hline
    Step~3, function evaluation & $s_{i}=\pi_{i}(s_0)$ & $e^{(0)}$ & $e^{(0)}$ & $e^{(0)}$ & $e^{(1)}$ & $e^{(0)}$ \\ \hline
    Step~4, output & $o(s_{i-1}, x_i)$ & - & $0$ & $0$ & $0$ & $1$ \\
    \hline
  \end{tabular}
  \\\vspace{.2cm}
  \begin{tabular}{|c|c|c c c c c|}
    \hline
    & $i$ & 0 & 1 & 2 & 3 & 4 \\ \hline
    input & $x_i$ & - & $0$ & $\mfu$ & $1$ & $0$ \\ \hline
    Step~1, function encoding & $\mathcal{M}_{t_{x_i}}$ & - & $\mathcal{M}_{t_0}$ & $\mathcal{M}_{t_{\mfu}}$ &
    $\mathcal{M}_{t_1}$ & $\mathcal{M}_{t_0}$\\ \hline
    Step~2, function composition & $\mathcal{M}_{\pi_{i}}$ & - & $\mathcal{M}_{t_0}$ & $\mathcal{M}_{t_{\mfu}}$ &
    \cellcolor{Gray} $\mathcal{M}_{t_\mfu}$ & \cellcolor{Gray} $\mathcal{M}_{t_\mfu}$\\ \hline
    Step~3, function evaluation & $s_{i}=\pi_{i}(s_0)$ & $e^{(0)}$ & $e^{(0)}$ & $\binom{\mfu}{\mfu}$ & \cellcolor{Gray}
    $\binom{0}{\mfu}$ & \cellcolor{Gray}$\binom{0}{\mfu}$\\ \hline
    Step~4, output & $o(s_{i-1}, x_i)$ & - & $0$ & $0$ & $\mfu$ & \cellcolor{Gray} $\mfu$ \\
    \hline
  \end{tabular}
\end{table*}}

Now consider Table~\ref{tab:examplerun}, which breaks down the computation for
two input strings, first the stable input $0010$ and then an input containing a
single unstable bit, $0\mfu10$. A hazard-free circuit should output $0001$ in
the first case and $00\mfu1$ in the latter case.
However, multiplication of any function encoding with the all-$\mfu$ matrix results
again in the all-$\mfu$ matrix, such that any further step of function composition
will return $\mathcal{M}_{t_{\mfu}}$. Hence, for the second input, the hazard-free
extension of the established approach will compute $\mfu$ as the last symbol.

To identify the key issue, examine the sequence of matrices determined from the
input symbols, which represent the transition functions restricted to the
respective input bit. $\mathcal{M}_{t_0}$ will map any vector corresponding to a
stable state, i.e., each unit vector, to $e^{(0)}$. This reflects the fact that a $0$
is guaranteed to result in state $0$. Accordingly, multiplying
$\mathcal{M}_{t_0}$ with any matrix representing the transition function
restricted to a stable input symbol will result in $\mathcal{M}_{t_0}$: no
matter what happened to the state machine before, the state after receiving
input symbol $0$ is $0$ (represented by $e^{(0)}$).

On the other hand, $\mathcal{M}_{t_{\mfu}}$ is the ``correct'' representation
for input symbol $\mfu$: regardless of the previous state, the resolutions
$0$ and $1$ of input symbol $\mfu$ reach state $0$ or $1$ respectively, and
$\binom{1}{0}*\binom{0}{1}=\binom{\mfu}{\mfu}$. Unfortunately,
(the hazard-free extension of) Boolean matrix multiplication of any matrix with
the all-$\mfu$ matrix $\mathcal{M}_{t_{\mfu}}$ can never yield a matrix that is
composed of column unit vectors.
The circuit computes
\begin{align*}
  &\mathcal{M}_{t_1} \cdot_{\mfu} \mathcal{M}_{t_{\mfu}} \cdot_{\mfu}
  \mathcal{M}_{t_0}\cdot_{\mfu}e^{(s_0)} \\&=
  \begin{pmatrix}
    0 & 0 \\
    1 & 1
  \end{pmatrix}
  \cdot_{\mfu}
  \begin{pmatrix}
    \mfu & \mfu \\
    \mfu & \mfu
  \end{pmatrix}
  \cdot_{\mfu}
  \begin{pmatrix}
    1 & 1 \\
    0 & 0
  \end{pmatrix}
  \cdot_{\mfu}
  \begin{pmatrix}
    1\\
    0
  \end{pmatrix}\\
  &=
  \begin{pmatrix}
    0 & 0 \\
    \mfu & \mfu
  \end{pmatrix}
  \cdot_{\mfu}
  \begin{pmatrix}
    1 \\
    0
  \end{pmatrix} = \begin{pmatrix}
    0 \\
    \mfu
  \end{pmatrix}
\end{align*}
as the (encoding of) state $s_3$ in Step~3.
Thus, in Step~4 the circuit at its best effort can only output
\begin{align*}
o_{\mfu}\left(\binom{0}{\mfu},0\right)&=
o\left(\binom{0}{0},0\right)*o\left(\binom{0}{1},0\right)\\
&=o\left(\binom{0}{0},0\right)*1
=\begin{cases}
1 & \mbox{if }o\left(\binom{0}{0},0\right)=1\\
\mfu & \mbox{if }o\left(\binom{0}{0},0\right)=0.
\end{cases}
\end{align*}
This might give the false hope of escaping the problem by leveraging our
aforementioned freedom to choose $o\left(\binom{0}{0},0\right)$ by setting it to
$1$, but this is a red herring. If $o\left(\binom{0}{0},0\right)=1$, then the
input string $0\mfu 00$ forces the circuit to incorrectly output
$o_{\mfu}\left(\binom{\mfu}{0},0\right)=o\left(\binom{0}{0},0\right)*0=\mfu$ as
the final bit.

Intuitively, the main takeaway from this example is that encoding the
transition function as an $|S|\times|S|$ matrix is insufficient to keep track of
the set of reachable states. The crucial problem is that uncertainty about the
transducer's state can be removed (or reduced) by later input symbols.
In our example, we have a very simple case: any stable input symbol fully
determines the attained state, regardless of the previous state.

In our approach, we keep track of a strict subset of states $A\subset S$
the state machine could have reached when facing some inputs with some uncertainty,
such that we can infer the set of states $B\subset S$
(ideally, a singleton if the uncertainty has been completely masked) that can be
reached by the current state transition.

\subsection{Suitable Encoding of Transition Functions}\label{subsec:encoding}
Before formalizing our encoding, we provide some intuition, by using, again, our
toy example, the shift transducer. As we kept the example tiny, the number of
subsets of the state space, i.e., the power set of $S$, is small: the four possible
subsets are $\emptyset$, $\{0\}$, $\{1\}$ and $\{0,1\}$.
Already a single $\mfu$ input leads to the largest possible
uncertainty about the state of the transducer, hence we choose $k=1$ throughout the example.

\subsubsection*{Intuition}
To avoid the pitfall discussed in Section~\ref{subsec:ppc}, we now choose a
highly redundant matrix representation. Fix an input symbol $\sigma \in
\{0,1\}$. The corresponding $2^{|S|}\times 2^{|S|}$ Boolean matrix encodes for
each pair of sets $A,B\subseteq S$ whether for each state in $A$ receiving
$\sigma$ as the next input symbol will result in a state from $B$. Again, assuming
that a hazard-free circuit is used to compute the matrix representation, this
choice fully determines the matrix
$\mathcal{M}_{t_{\mfu}}=\mathcal{M}_{t_0}*\mathcal{M}_{t_1}$ resulting from
input symbol $\mfu$. Labeling the rows by subsets $B$ and columns by the subsets $A$,
the resulting matrices $\mathcal{M}_{t_0}$, $\mathcal{M}_{t_1}$,
and $\mathcal{M}_{t_{\mfu}}$ are
\begin{align*}
&\begin{blockarray}{ccccc}
    & \emptyset & \{0\} & \{1\} & \{0,1\}\\
    \begin{block}{c(cccc)}
      \emptyset & 1 & 0 & 0 & 0 \\
      \{0\}   & 1 & 1 & 1 & 1 \\
      \{1\}   & 1 & 0 & 0 & 0 \\
      \{0,1\} & 1 & 1 & 1 & 1 \\
    \end{block}
  \end{blockarray}~,~~\\
 & \begin{blockarray}{ccccc}
    & \emptyset & \{0\} & \{1\} & \{0,1\}\\
    \begin{block}{c(cccc)}
      \emptyset & 1 & 0 & 0 & 0 \\
      \{0\}   & 1 & 0 & 0 & 0 \\
      \{1\}   & 1 & 1 & 1 & 1 \\
      \{0,1\} & 1 & 1 & 1 & 1 \\
    \end{block}
  \end{blockarray}~,\mbox{ and}~~\\
  &\begin{blockarray}{ccccc}
    & \emptyset & \{0\} & \{1\} & \{0,1\}\\
    \begin{block}{c(cccc)}
      \emptyset & 1 & 0 & 0 & 0 \\
      \{0\}   & 1 & \mfu & \mfu & \mfu \\
      \{1\}   & 1 & \mfu & \mfu & \mfu \\
      \{0,1\} & 1 & 1 & 1 & 1 \\
    \end{block}
  \end{blockarray}~,
\end{align*}
respectively.
Consider, for example, input symbol $0$ and the corresponding matrix $\mathcal{M}_{t_0}$.
Each set of states $\{0\}$,
$\{1\}$ and $\{0,1\}$ will transition to state $0$. As $0$ is a subset of $\{0\}$ and $\{0,1\}$
the respective column vectors of the matrix are $\begin{pmatrix}0 & 1 & 0 & 1\end{pmatrix}^{\intercal}$.
Note that each matrix maintains the
trivialities that the empty set will always be mapped to a subset of any set
(leftmost column), no non-empty set is mapped to a subset of the empty set (top
row), and any set will be mapped to a subset of $S=\{0,1\}$ (bottom
row).\footnote{Note that the submatrices induced by the rows and columns of
singleton sets equal those we got in Section~\ref{subsec:ppc}. Since we opted
for a minimal example with only two states, none of the additional entries
depend on the specific transition function. Note that this changes for $|S|>2$.}
Crucial to us is the point that the encoding now ``takes note'' of the fact that
even when an input symbol is $\mfu$, it remains certain that for any resolution
of the input the transducer must end up in \emph{some} state, reflected by the
bottom row of $\mathcal{M}_{t_{\mfu}}$.

\subsubsection*{Application to the example}
Applying the framework of Section~\ref{subsec:ppc} with the new encoding to the
input string $0\mfu10$, this time Step~2 yields for $\pi_{3}$:
\begin{align*}
  &\mathcal{M}_{t_1} \cdot_{\mfu} \mathcal{M}_{t_{\mfu}} \cdot_{\mfu}
  \mathcal{M}_{t_0}\\
  &=
  \begin{pmatrix}
    1 & 0 & 0 & 0 \\
    1 & 0 & 0 & 0 \\
    1 & 1 & 1 & 1 \\
    1 & 1 & 1 & 1
  \end{pmatrix}
  \cdot_{\mfu}
  \begin{pmatrix}
    1 & 0 & 0 & 0 \\
    1 & \mfu & \mfu & \mfu \\
    1 & \mfu & \mfu & \mfu \\
    1 & 1 & 1 & 1
  \end{pmatrix}
  \cdot_{\mfu}
  \begin{pmatrix}
    1 & 0 & 0 & 0 \\
    1 & 1 & 1 & 1 \\
    1 & 0 & 0 & 0 \\
    1 & 1 & 1 & 1
  \end{pmatrix}\\
  &=
  \begin{pmatrix}
    1 & 0 & 0 & 0 \\
    1 & 0 & 0 & 0 \\
    1 & 1 & 1 & 1 \\
    1 & 1 & 1 & 1
  \end{pmatrix}
  \cdot_{\mfu}
  \begin{pmatrix}
    1 & 0 & 0 & 0 \\
    1 & \mfu & \mfu & \mfu \\
    1 & \mfu & \mfu & \mfu \\
    1 & 1 & 1 & 1
  \end{pmatrix}
%  &=
%  \begin{pmatrix}
%    1 & 0 & 0 & 0 \\
%    1 & 0 & 0 & 0 \\
%    1 & 1 & 1 & 1 \\
%    1 & 1 & 1 & 1
%  \end{pmatrix}
  =\mathcal{M}_{t_1}.
\end{align*}
As we can see, multiplying with $\mathcal{M}_{t_1}$ from the left now correctly
recovers $\mathcal{M}_{t_1}$, i.e., regardless of previous possibly unstable
input symbols, the computed matrix reflects that reading input symbol $1$
results in state $1$. The above behavior is not a corner case due to the small size of our example,
a larger example is given in Section~\ref{sec:example_sec}.

A fundamental problem in hazard-free circuits is that the resolution
of the superposition may add undesired values (Observation~\ref{obs:resstar}).
Recall that for a Boolean function $f\colon \IB^n \to \IB^m$, we obtain         $f_{\mfu}(x)$ by mapping each $y\in \res(x)$ using $f$ and then taking the
$\bigstarr$ operation over the resulting set. The latter might, depending on $x$
and the encoding, lose information, as $\res(f_{\mfu}(x))$ might be a strict
superset of $f(\res(x))$. This becomes problematic when we subsequently apply
some function, e.g., $g:\IB^m \to \IB$, that is constant on $f(\res(x))$, but not on
$\res(f_{\mfu}(x))$ for some $x\in \IT^n$; we then get that
\begin{align*}
\mfu&=g_{\mfu}(f_{\mfu}(x))=\bigstarr(g(\res(f_{\mfu}(x))))\\
&\neq\bigstarr(g(f(\res(x)))=(g\circ f)_{\mfu}(x)\in \IB\,.
\end{align*}

\subsubsection*{A universal encoding for functions}
The key idea underlying our solution to this
problem is to maintain the information that $f$ maps $\res(x)$ to $f(\res(x))$.
As illustrated by the example, the encoding stores for each $A\subseteq \IB^n$ and $B\subseteq \IB^m$
whether $f(A)\subseteq B$. When composing functions, we then can retrieve
the information that $g\circ f$ is constant on $\res(x)$.

The size of the encoding can be reduced if we only regard $k$-bit hazards.
If the number of $\mfu$'s in the input $x$ can be bounded by an integer $k$, then
there is also an upper bound on $|f(\res(x))|$. As each $\mfu$ has
two stable resolutions, we can readily bound $|f(\res(x))|\le |\res(x)|\le 2^k$.
Hence, the encoding can be reduced to sets $A\subseteq \IB^n$ and $B\subseteq \IB^m$,
where $|A|\leq2^k$ and $|B|\leq2^k$.

This leads to the
following encoding, which is universal in the sense that it gives rise to
$k$-bit hazard-free implementations of arbitrary transducers.

\begin{definition}[Universal Function Encoding]\label{def:universal}
Denote by $\mathcal{P}_{t}(A)$ the set of all subsets of $A$ with cardinality
smaller equal to $t\in\{0,\ldots,|A|\}$, i.e., $\mathcal{P}_{t}(A)\coloneqq\{A'\subseteq
A\,|\,|A'|\leq t\}$.
Given a function $f\colon S\to T$ and $k\in\NN$, define
\begin{equation*}
\forall A\in\encset{S}, B\in\encset{T}\colon \M{f}_{BA}\coloneqq \begin{cases}
1 & \mbox{if }f(A)\subseteq B\\
0 & \mbox{else.}
\end{cases}
\end{equation*}
Thus, the Boolean matrix $\mathcal{M}_{f}$ has dimension
\[|\encset{T}|\times|\encset{S}|=
\sum^{\min\{|T|, 2^k\}}_{i=0}\begin{pmatrix}{|T|}\\i\end{pmatrix}\times
\sum^{\min\{|S|, 2^k\}}_{i=0}\begin{pmatrix}{|S|}\\i\end{pmatrix}.\]
Moreover,
for $s\in S$ and $A\in\encset{S}$, define $e^{(s)}$ via $e_A^{(s)}\coloneqq 1$ if
$s\in A$ and $e_A^{(s)}\coloneqq 0$ otherwise. Hence, for all $B\in\encset{T}$ we have
that $\left(\mathcal{M}_{f}\cdot e^{(s)}\right)_B=1$ if $f(s)\in
B$ and $\left(\mathcal{M}_{f}\cdot e^{(s)}\right)_B=0$ otherwise.
\end{definition}

We remark that we are mostly interested in the case where $T=S$,
since for the restricted transition functions computed in Step~1 we only need to
represent functions from $S$ to $S$.

\subsection{Proving the Main Result}\label{subsec:poutline}
Our goal in this subsection is to show Theorem~\ref{thm:trans}. To this end, we
first establish that the above encoding indeed keeps track of all required
information to remove uncertainty in case the input allows it.
We show that the above matrix representation is a suitable
encoding, i.e., we show that the representation is capable of encoding the transition
function without dropping information, here the transition function is
restricted to a single input symbol.

\subsubsection*{The PPC framework of Ladner and Fischer}
Recall that the PPC framework computes states $s_i$ by
(Step~1) translating input symbol $x_i$ into the matrix representation of $t_{x_i}$,
(Step~2) determining by matrix multiplication the transition function $\pi_i$ resulting
from a sequence of input symbols, and
(Step~3) evaluating the transition function $\pi_i$ on $e^{(s_0)}$ via matrix-vector multiplication.
Given state $s_i$ the output at position $i$ can be determined by application of $o_{x_i}$ (Step~4).
In the PPC framework, we replace the original matrix representation
with the new universal function encoding.
We show that the universal encoding overcomes the issue of information
loss during function composition.
If function composition does not lose information, then repeated application of
function composition gives hazard-free transition functions $\pi_i$,
which will be formalized in Corollary~\ref{cor:extension_commutes}.

\subsubsection*{Main stepping stone}
To show that for
the universal encoding the strategy also succeeds in the face of unstable inputs,
we need to prove that composing functions and translating the composed function
into its matrix representation is equivalent to first translating each function
to its matrix representation and then multiplying these matrices. This is
captured by the following theorem, which is our main stepping stone towards
Theorem~\ref{thm:trans}.

\begin{theorem}\label{thm:extension_commutes}
Let $k\in\NN$, $f_j\colon S\to T$ for all $j\in J$, $g_i\colon T\to U$ for all $i\in
I$, $A\in\encset{S}$, and $C\in\encset{U}$. If $|J|\cdot|A|\leq2^k$, then
\begin{equation*}
\left(\left(\bigstarr_{i\in I} \mathcal{M}_{g_i}\right)\cdot_{\mfu}
\left(\bigstarr_{j\in J} \mathcal{M}_{f_j}\right)\right)_{CA}
=\left(\bigstarr_{(i,j)\in I\times J} \mathcal{M}_{g_i\circ
f_j}\right)_{CA}.
\end{equation*}
\end{theorem}

The condition $|J|\cdot|A|\leq2^k$ may seem non-intuitive at first.
The product $|J|\cdot|A|$ corresponds to the number of resolutions of the input that
has been processed so far.
The product is bounded by $2^k$, i.e., the number of resolutions of $k$ many $\mfu$'s.
In the application of Theorem~\ref{thm:extension_commutes}, set $J$ corresponds
to the resolutions of respective parts of the input.
Set $A$ corresponds to the current state of the transducer, and its size depends
on the uncertainty of previous transitions.
More intuition is provided in the example in Section~\ref{sec:example_sec}.

Before we prove the key stepping stone we discuss tools that are used in
the proof of the main result and the key stepping stone. First, we define hazard-free
multiplexers which are used in Step~$4$ of the PPC framework.
Second, we show that there is an efficient implementation of hazard-free matrix
multiplication. Third, we introduce monotone resolutions, a technique used in the
proofs. Last, we state a recent result on the complexity
of hazard-free circuits for general functions, which is applied in the proof of
the main theorem.

\subsubsection*{Hazard-free multiplexer}
For later use, we define the $\ell$-input multiplexer $\MUX_\ell$. A multiplexer
is a circuit that selects one of its inputs according to a dedicated select
input. The select input encodes index $i\in[2^\ell]$, where $[t]=\{0,\ldots,t-1\}$, for $t\in\NN_{>0}$.

\begin{definition}\label{def:cmux}
  Let $\ell,b\in\NN_{>0}$. An $\ell$-input multiplexer $\MUX_\ell$ receives inputs
  $x_i\in\IB^b$ for $i\in[2^\ell]$, and select input $s\in\IB^\ell$.
  Let $\langle \cdot \rangle\colon\IB^\ell\rightarrow\NN$ be the standard binary decoding function.
  Interpreting the select input $s$ as an index,
  $\MUX_\ell$ outputs $x_{\langle s\rangle}$, i.e., $\MUX_\ell(x_0,\ldots,x_{2^\ell-1},s)\coloneqq
  x_{\langle s\rangle}$.
\end{definition}

Hazard-free multiplexers are an efficient tool for selecting from a set of functions
encoded by the universal function encoding.
For ease of readability, we name these multiplexers $\MMUX$.
Given select input $s\in \IB^\ell$ and a set of functions $\{f_j:S\rightarrow T | j\in \IB^\ell\}$,
interpret $j$ as binary numbers, set all $\mathcal{M}_{f_j}$ as the multiplexer
inputs, and $s$ as the select input. Then $\MMUX$ selects the universal encoding
of $f_j$, where $\langle j\rangle=\langle s\rangle$;
\begin{align*}
  \MMUX_\ell(\{f_j|j\in \IB^\ell\},s)\coloneqq
\mathcal{M}_{f_{s}}\,.
\end{align*}

The $\ell$-input multiplexer
can be implemented by a hazard-free circuit~\cite{ikenmeyer18complexity}.

\begin{corollary}[of~{\cite[Lemma 5.1]{ikenmeyer18complexity}}]\label{col:cmux}
  Let $\ell,b\in\NN_{>0}$, there is a hazard-free implementation of $\MUX_\ell$
  computing $\left(\MUX_\ell\right)_\mfu(x_0,\ldots,x_{2^\ell-1},s)$ with $x_i\in\IB^b$
  for $i\in[2^\ell]$ and $s\in\IB^\ell$.
  The implementation has
  \begin{align*}
    \text{size}&\quad \BO\left(2^\ell b\right)\\
    \text{and depth}&\quad\BO\left(\ell\right)\,.
  \end{align*}
\end{corollary}
\begin{IEEEproof}
  We only sketch the proof as this is a standard construction.
  The hazard-free implementation of a multiplexer receiving two input bits
  and a single select bit is given in~\cite{ikenmeyer18complexity}; it has constant
  size and depth.
  The multiplexer can be extended to $\ell$ select bits by building a tree of
  multiplexers, where each layer is controlled by a bit of the select input.
  Extension to inputs of width $b$ is simply done by copying the tree of multiplexers
  $b$ times.
\end{IEEEproof}

\subsubsection*{Hazard-free matrix multiplication}
For Theorem~\ref{thm:extension_commutes} to be of use, we need a circuit
implementing $\cdot_{\mfu}$, i.e., hazard-free matrix multiplication. The
standard Boolean matrix multiplication algorithm is known to be appropriate.
\begin{corollary}[of~{\cite[Lemma~4.2]{ikenmeyer18complexity}}]\label{cor:monotone}
There is a circuit of size $(2\beta-1)\alpha\gamma$ and depth
$\lceil\log \beta\rceil+1$ that computes $\mathcal{A}\cdot_{\mfu}\mathcal{B}$ for matrices $\mathcal{A}\in
\IT^{\alpha\times \beta}$ and $\mathcal{B}\in \IT^{\beta \times \gamma}$.
\end{corollary}
\begin{IEEEproof}
The standard algorithm for Boolean matrix multiplication is monotone, i.e., does
not use negations, and requires for each of the $\alpha \gamma$ entries of
$\mathcal{A}\cdot_{\mfu}\mathcal{B}$ a binary tree of $\beta$ $\AND$ gates (the leaves) and
$\beta-1$ $\OR$ gates (internal nodes); monotone circuits are hazard-free.
\end{IEEEproof}

We observe that hazard-free Boolean matrix multiplication is associative.

\begin{observation}[$\cdot_{\mfu}$ is associative]\label{obs:matrix_associative}
For all $A\in \IT^{\alpha \times \beta}$, $B\in \IT^{\beta \times \gamma}$, and
$C\in \IT^{\gamma \times \delta}$, we have that
$(A\cdot_{\mfu}B)\cdot_{\mfu}C=A\cdot_{\mfu}(B\cdot_{\mfu}C)$.
\end{observation}
\begin{IEEEproof}
As $\OR$ and $\AND$ are associative also on $\IT$, this follows by the same,
straightforward calculation for matrices over arbitrary
(semi)rings.\footnote{Note that $(\IT,\OR,\AND)$ is only a (commutative)
semiring, as its ``addition,'' i.e., $\OR$, has no inverses.}
\end{IEEEproof}

To prove Theorem~\ref{thm:extension_commutes}, we first need to establish that
matrix multiplication indeed is equivalent to function composition for stable
inputs.
\begin{lemma}\label{lem:matrix_stable}
Let $f$ and $g$ be functions $f\colon S\to T$ and $g\colon T\to U$.
For all $A\in\mathcal{P}_{2^k}(S)$ and $C\in\mathcal{P}_{2^k}(U)$, it holds that
$\left(\mathcal{M}_g\cdot \mathcal{M}_f\right)_{CA}=\left(\mathcal{M}_{g\circ f}\right)_{CA}$.
\end{lemma}
\begin{IEEEproof}
Suppose $\M{g\circ f}_{CA}=1$, i.e., $(g\circ f)(A)=g(f(A))\subseteq C$. Note that $A\in\encset{S}\Rightarrow f(A)\in\encset{T}$. Then,
\begin{align*}
\left(\M{g}\cdot \M{f}\right)_{CA} &= \sum_{B\in\mathcal{P}_{2^k}(T)}
\M{g}_{CB} \M{f}_{BA}\\
&\geq \M{g}_{Cf(A)}\M{f}_{f(A)A}\\
&= 1\cdot 1
= 1\,,
\end{align*}
where we use that $\M{f}_{f(A)A} = 1$ by Definition~\ref{def:universal}
and $\M{g}_{Cf(A)}=1$ because $g(f(A))\subseteq C$.

Now consider the case that $\M{g\circ f}_{CA}=0$, i.e., there exists $a\in A$ so that $g(f(a))\not\subseteq C$.
Accordingly,
\begin{equation*}
\forall B\in\mathcal{P}_{2^k}(T)\colon \M{f}_{BA}=1 \Rightarrow \M{g}_{CB}=0\,,
\end{equation*}
because
\begin{equation*}
\M{f}_{BA}=1 \Leftrightarrow f(A)\subseteq B
\Rightarrow g(B)\not \subseteq C\Leftrightarrow \M{g}_{CB}=0\,.
\end{equation*}
Thus, for all $B\in\mathcal{P}_{2^k}(T)$, we have that $\M{f}_{BA}\M{g}_{CB}=0$, leading to
\begin{displaymath}
\left(\mathcal{M}_{g}\cdot \mathcal{M}_{f}\right)_{CA} = \sum_{B\subseteq T}
\M{g}_{CB}\M{f}_{BA} = 0\:.\IEEEQEDhereeqn
\end{displaymath}
\end{IEEEproof}

\subsubsection*{Monotone resolutions}
To understand how multiplying matrices plays out when the multiplicands are
not stable, we exploit the monotonicity of matrix multiplication. Because
flipping matrix entries of multiplicands from $0$ to $1$ can only flip entries
from $0$ to $1$ in the product, we can restrict our attention to only two
resolutions of each matrix: we simultaneously replace all $\mfu$ entries by either $0$ or $1$,
respectively.
\begin{definition}\label{def:matrix_simple}
For $A\in \IT^{\alpha \times \beta}$ and $b\in \IB$, define $A^{(b)}\in
\IB^{\alpha\times \beta}$ via                                                   \begin{align*}
\forall (i,j)\in \{1,\ldots,\alpha\}\times \{1,\ldots,\beta\}\colon A^{(b)}_{ij}\coloneqq
\begin{cases} b & \mbox{if }A_{ij}=\mfu\\
A_{ij} & \mbox{else}\,.
\end{cases}
\end{align*}
\end{definition}
With this definition, the above intuition is formalized by the following lemma.
\begin{lemma}\label{lem:matrix_simple}
For all $G\in \IT^{\alpha\times \beta}$, $F\in \IT^{\beta\times \gamma}$ and all $i\in\{1,\ldots,\alpha\}$, $j\in\{1,\ldots,\gamma\}$, we have
\[\left(G\cdot_{\mfu}F\right)_{ij} = \mfu \Leftrightarrow
\left(G^{(0)}\cdot F^{(0)}\right)_{ij} = 0 \land \left(G^{(1)}\cdot F^{(1)}\right)_{ij} = 1.\]
\end{lemma}
\begin{IEEEproof}
Fix $i\in \{1,\ldots,\alpha\}$ and $j\in \{1,\ldots,\gamma\}$. If
$\left(G\cdot_{\mfu}F\right)_{ij}=b\in \IB$, i.e., for all $G',F'$ such that
$G'\in\res(G)$ and $F'\in\res(F)$, $\left(G'\cdot F'\right)_{ij}=b$, then since $G^{(0)},G^{(1)}\in
\res(G)$ and $F^{(0)},F^{(1)}\in \res(F)$ it holds that that \[\left(G^{(0)}\cdot F^{(0)}\right)_{ij} = \left(G^{(1)}\cdot F^{(1)}\right)_{ij} = b.\]
Now consider the case that $\left(G\cdot_{\mfu}F\right)_{ij}=\mfu$. Thus,
there are $G',G''\in \res(G)$ and $F',F''\in \res(F)$ satisfying that
$\left(G'\cdot F'\right)_{ij}=0$ and $\left(G''\cdot F''\right)_{ij}=1$,
respectively. It follows that
\begin{align*}
\left(G^{(0)}\cdot F^{(0)}\right)_{ij}
&= \sum_{k=1}^{\beta} G^{(0)}_{ik}F^{(0)}_{kj}\\
&\leq \sum_{k=1}^{\beta} G'_{ik}F'_{kj}
=\left(G'\cdot F'\right)_{ij}
= 0
\end{align*}
and, analogously,
\begin{displaymath}
\left(G^{(1)}\cdot F^{(1)}\right)_{ij}\geq \left(G''\cdot F''\right)_{ij}=1\,.\IEEEQEDhereeqn
\end{displaymath}
\end{IEEEproof}
\subsubsection*{Hazard-free circuit complexity}
Jukna~\cite{jukna2021notes} presents an upper bound
on the complexity of the hazard-free implementation of a Boolean function $f$.
\begin{theorem}[\cite{jukna2021notes}]\label{thm:jukna}
  Given a Boolean function $f\colon\IB^n\rightarrow\IB$, there is a hazard-free
  circuit implementing $f$ that has size $\BO(2^n/n)$ and depth $\BO(n)$.
\end{theorem}
The size bound of the implementation is shown in \cite{jukna2021notes}. The depth
bound follows by the following argument.
The hazard-free construction consists of two consecutive
parts that are recursively defined.
The first part uses at most $m$ recursions and the second part uses $n-m$ recursions of the same recursion step.
As each recursion step has a constant size, the construction has depth $\BO(n)$.

\subsubsection*{Proving the key stepping stone}
Using Lemma~\ref{lem:matrix_simple}, proving Theorem~\ref{thm:extension_commutes} is reduced to
showing correct behavior for matrices $\left(\bigstarr_{j\in J} \mathcal{M}_{f_j}\right)^{(0)}$
and $\left(\bigstarr_{j\in J} \mathcal{M}_{f_j}\right)^{(1)}$ instead of all
resolutions of $\bigstarr_{j\in J} \mathcal{M}_{f_j}$ and $\bigstarr_{i\in I} \mathcal{M}_{g_i}$.
\begin{IEEEproof}[Proof of Theorem~\ref{thm:extension_commutes}]
  Define $\preceq$ as the partial order $b\prec \mfu$ for $b\in \IB$ and observe
  that $\bigstarr X \preceq \bigstarr Y$ for $X\subseteq Y\subseteq \IB$.
  By Observation~\ref{obs:resstar}, we obtain for $\mathcal{M}_{g_i}$ (and accordingly $\mathcal{M}_{f_j}$) that
  \begin{align*}
    \{\mathcal{M}_{g_i}|i\in I\}\subseteq \res\left(\bigstarr_{i\in I}\mathcal{M}_{g_i}\right)\,.
  \end{align*}
  Thus, by the definition of the hazard-free extension (Definition~\ref{def:hfree})
  and the resolution (Definition~\ref{def:resolution}),
  \begin{align*}
    &\left(\left(\bigstarr_{i\in I} \mathcal{M}_{g_i}\right)\cdot_{\mfu}
    \left(\bigstarr_{j\in J} \mathcal{M}_{f_j}\right)\right)_{CA}\\
    &=\bigstarr \left(\res\left(\bigstarr_{i\in I}
    \mathcal{M}_{g_i}\right)\cdot \res\left(\bigstarr_{j\in J}
    \mathcal{M}_{f_j}\right)\right)_{CA}\\
    & \succeq \bigstarr\left(\{\mathcal{M}_{g_i}|i\in I\}\cdot \{\mathcal{M}_{f_j}|j\in J\}\right)_{CA}\\
    & = \left(\bigstarr_{(i,j)\in
    I\times J} \mathcal{M}_{g_i}\cdot \mathcal{M}_{f_j}\right)_{CA}\\
    &= \left(\bigstarr_{(i,j)\in
    I\times J} \mathcal{M}_{g_i\circ f_j}\right)_{CA}\,,
  \end{align*}
  where the last equality follows from Lemma~\ref{lem:matrix_stable}.
  The claimed equality follows if the l.h.s.\ equals $b\in \{0,1\}$.

  It remains to show the claimed equality assuming that the l.h.s.\ equals $\mfu$.
  By application of Lemma~\ref{lem:matrix_simple} with $G=\bigstarr_{i\in I}
  \mathcal{M}_{g_i}$ and $F=\bigstarr_{j\in J} \mathcal{M}_{f_j}$, we obtain that
  \begin{equation*}
    \left(\left(G^{(0)}\cdot F^{(0)}\right)_{CA}=0\right) \wedge \left(\left(G^{(1)}\cdot F^{(1)}\right)_{CA}=1\right)\,.
  \end{equation*}
  By the definition of matrix multiplication, this is equivalent to
  \begin{align}
    \forall B\in \mathcal{P}_{2^k}(T)&\colon G^{(0)}_{CB}=0 \lor F^{(0)}_{BA}=0\,, \label{eq:allzero}\\
    \exists B\in \mathcal{P}_{2^k}(T)&\colon G^{(1)}_{CB}=1\land F^{(1)}_{BA}=1\,. \label{eq:existsone}
  \end{align}
  We observe from Definition~\ref{def:matrix_simple} that for $b\in\IB$ we have that
  \begin{align*}
    \left(\bigstarr_{i\in I}\mathcal{M}_{g_i}\right)^{(b)}_{CB}=b\Leftrightarrow
    \exists i\in I\colon\left(\mathcal{M}_{g_i}\right)_{CB}=b\,;
  \end{align*}
  an analogous statement holds for $\mathcal{M}_{f_j}$.
  Plugging this observation into equations \eqref{eq:allzero} and \eqref{eq:existsone}, we get that
  \begin{align*}
    \forall B\in\mathcal{P}_{2^k}(T)\; \exists (i,j)\in I\times J&\colon\\
    \M{g_i}_{CB}=0&\lor \M{f_j}_{BA}=0\,, \label{eq:def1}\tag{3}
  \end{align*}
  and
  \begin{align*}
    \exists B\in\mathcal{P}_{2^k}(T)\; \exists (i,j)\in I\times J&\colon\\
    \M{g_i}_{CB}=1&\land\M{f_j}_{BA}=1\label{eq:def2}\,.\tag{4}
  \end{align*}
  Let $B_0=\bigcup_{j\in J}f_j(A)$ be the subset of $T$, to which states in $A$ are mapped to
  by any $f_j$. As $|f_j(A)|\leq |A|$, cardinality $|B_0|$ is at most $|J|\cdot|A|$.
  By assumption $|J|\cdot|A|\leq2^k$, we obtain $B_0\in\encset{T}$. Since $f_j(A)\subseteq B_0$ by
  construction, it holds that $\M{f_j}_{B_0 A}=1$ for all $j\in J$.
  Equation~\eqref{eq:def1} thus entails that
  \begin{equation*}
    \exists i\in I\colon \M{g_i}_{C B_0}=0 \Leftrightarrow \exists i\in I\colon g_i(B_0)\not
    \subseteq C\,.
  \end{equation*}
  Hence, there are $i_0\in I$ and $x\in B_0$ such that $g_{i_0}(x)\notin C$. By construction, $x\in f_{j_0}(A)$ for some $j_0\in J$,
  yielding that $(g_{i_0}\circ f_{j_0})(A)=g_{i_0}(f_{j_0}(A))\not \subseteq C$.
  We conclude that $\M{g_{i_0}\circ f_{j_0}}_{CA}=0$.

  Now consider equation~\eqref{eq:def2}, which says that there are indices  $i_1\in I$ and $j_1\in
  J$ such that $g_{i_1}(B_1)\subseteq C$ and
  $f_{j_1}(A)\subseteq B_1$. This immediately yields that $(g_{i_1}\circ
  f_{j_1})(A)\subseteq C$ and thus $\M{g_{i_1}\circ f_{j_1}}_{CA}=1$.

  The desired equality now follows, because
  \begin{align*}
    &\left(\bigstarr_{(i,j)\in I\times J} \mathcal{M}_{g_i\circ f_j}\right)_{CA}\\&\succeq
    \left(\bigstarr\{\mathcal{M}_{g_{i_0}\circ f_{j_0}}, \mathcal{M}_{g_{i_1}\circ f_{j_1}}\}\right)_{CA}\\
    &= \bigstarr\{\M{g_{i_0}\circ f_{j_0}}_{CA}, \M{g_{i_1}\circ f_{j_1}}_{CA}\}\\
    &= \bigstarr\{0,1\} = \mfu
  \end{align*}
  and $b\succeq\mfu$ only holds if $b=\mfu$.
\end{IEEEproof}

With Theorem~\ref{thm:extension_commutes} at our disposal, we are ready to prove our main result,
Theorem~\ref{thm:trans}. Following Step~1 and Step~2 of the parallel prefix framework
given in Section~\ref{subsec:ppc}, we need to compute $\pi_i=t_{x_{i}} \circ \hdots \circ t_{x_1}$
for all prefixes $x_i\hdots x_1$ of the input string.
This computation can be phrased in terms of matrix multiplications, which is shown
by the following corollary. It readily follows by inductive application of
Theorem~\ref{thm:extension_commutes}.
\begin{corollary}\label{cor:extension_commutes}
Suppose that for $i\in [n]$, we are given mappings $E_i\colon \IB^{\ell}\to F_i$
from input symbols $\IB^{\ell}$ to function spaces $F_i$. Moreover, for all $i\in [n-1]$ the codomain of
functions from $F_i$ equals the domain of functions from $F_{i+1}$. Let
$E\colon \IB^{n\ell}\to(F_0\to F_{n-1})$ denote a function that maps a binary string $x\in\IB^{n\ell}$
to the composition of the corresponding functions, $E(x)\coloneqq\circ_{i=0}^{n-1} E_i(x_i)$.
Then, for all $x\in\IT^{n\ell}$, \begin{align*}
\M{E(\cdot)}_{\mfu}(x)&=\M{E_{n-1}(\cdot)}_{\mfu}(x_{n-1})\cdot_{\mfu}
\M{E_{n-2}(\cdot)}_{\mfu}(x_{n-2})\\
&\quad\quad\cdot_{\mfu}\ldots\cdot_{\mfu}\M{E_0(\cdot)}_{\mfu}(x_0)\,.
\end{align*}
\end{corollary}

Following this insight we observe that the evaluation of the functions corresponds
to hazard-free matrix-vector multiplication.

\begin{corollary}\label{cor:faithful}
  For $j\in J$, let $f_j\colon S\to T$. Assume that $A\in \encset{T}$, $S'\in\encset{S}$, and $|J|\cdot|S'|\leq k$. Then
  \begin{equation*}
    \left(\left(\bigstarr_{j\in J}\mathcal{M}_{f_j}\right)\cdot_\mfu\left(\bigstarr_{s\in S'}e^{(s)}\right)\right)_{A} =
    \left(\bigstarr_{(j,s)\in J\times S'}e^{(f_j(s))}\right)_{A}\,.
  \end{equation*}
\end{corollary}
\begin{IEEEproof}
  Define $g_s\colon \{\bullet\}\rightarrow S$ by $g_s(\bullet)\coloneqq s$ for $s\in S$,
  such that $\M{g_s}_{\{\bullet\}A}=\left(e^{(s)}\right)_A$.
  By Theorem~\ref{thm:extension_commutes}, we thus get that
  \begin{align*}
    &\left(\left(\bigstarr_{j\in J}\mathcal{M}_{f_j}\right)\cdot_\mfu\left(\bigstarr_{s\in S'}e^{(s)}\right)\right)_{A}\\
    &= \left(\left(\bigstarr_{j\in J}\mathcal{M}_{f_j}\right)\cdot_\mfu
    \left(\bigstarr_{s\in S'}\mathcal{M}_{g_s}\right)\right)_{\{\bullet\}A}\\
    &= \left(\bigstarr_{(j,s)\in J\times S'}\mathcal{M}_{f_j\circ g_s}\right)_{\{\bullet\}A}
\end{align*}
  The claim follows by applying the definition of $g_s$ again.
\end{IEEEproof}

\subsubsection*{Multiplication of all prefixes}
Corollary~\ref{cor:extension_commutes} uses the hazard-free matrix product of all input prefixes. This
can be efficiently implemented, similar to the parallel prefix computation in the
approach of Ladner and Fischer.
\begin{corollary}[of~{\cite[Section 2]{ladner1980parallel}}]\label{col:ppc}
  For input matrices $\mathcal{A}_{n-1}, \dots, \mathcal{A}_0$ of size $\alpha\times\alpha$,
  there is a circuit computing the hazard-free Boolean matrix multiplication of all prefixes;
  $\mathcal{A}_i\cdot_\mfu\ldots\cdot_\mfu\mathcal{A}_0$,
for each $i\in[n]$. The circuit
  \begin{align*}
    \text{has size}&\quad \BO\left(\alpha^3n\right)\\
    \text{and depth}&\quad \BO\left(\log\alpha\log n \right)\,.
  \end{align*}
\end{corollary}
\begin{IEEEproof}
  By Observation~\ref{obs:matrix_associative}, $\cdot_{\mfu}$ is associative.
  For an associative operator, Ladner and Fischer~\cite{ladner1980parallel} present a family of
  circuits computing the application of all prefixes of the input.
  Let $c$ be the size and $d$ the depth of a circuit implementing the operator. The family has
  asymptotically optimal size $\BO(cn)$ and depth $\BO(d\log n)$.
  Corollary~\ref{cor:monotone} offers an implementation of hazard-free Boolean $\alpha\times\alpha$
  matrix multiplication of size $c=\alpha^3$ and depth $d=\log\alpha$.
\end{IEEEproof}

\subsubsection*{Output step}
Before putting the above pieces together to derive our main results, we need to
address how the final output is computed, i.e., Step~4. Here, we can exploit
that (i) the output does not need to be represented in the universal encoding, and
(ii) the universal encoding used in the previous computations holds additional
information that simplifies determining the output in a hazard-free way.
We leverage these points in a case analysis to minimize the cost
of the output stage.

The input to Step~4 is the vector encoding state $s_{i-1}$ and input $x_i$ for
each $i\in\{1\ldots n\}$; it computes output $o(s_{i-1},x_i)$. We restrict the
output function to an input symbol, as we did for the transition function, i.e.,
$o_\sigma\colon S\rightarrow\Lambda$ for $\sigma\in\Sigma$ is defined by
$o_\sigma(s)\coloneqq o(s,\sigma)$.
In what follows we describe the computation of output bit $o_\sigma(s)_j$ ($=o(s,\sigma)_j$),
for $j\in[m]$.
We remark that the preimage of $1$ under
$o_\sigma(s)_j$ is a set of states. Recall that by Definition~\ref{def:universal},
the state vector $e^{(s_{i-1})}$ encodes not only state $s_{i-1}$, but indicates
for each subset of states $S'\subseteq S$ (with cardinality less or equal to
$2^k$) whether $s_{i-1}\in S'$.
Thus, we can simply check whether $s_{i-1}$ lies in the set that of states which
$o_\sigma(s_{i-1})_j$ maps to $1$ (provided its cardinality is at most $2^k$).
\begin{definition}\label{def:preimage}
  For an input symbol $\sigma\in\Sigma$ and $j\in[m]$, we define
  \begin{equation*}
    \Asigj \coloneqq \{s\in S | o(s,\sigma)_j=1\}\,.
  \end{equation*}
\end{definition}

Note that if $2^k < |S|$, we reduce the size of the universal encoding by
encoding only sets of size less or equal to $2^k$. Hence, we might not
have computed a bit indicating whether $s_{i-1}\in \Asigj$ is an entry of the
vector $e^{(s_{i-1})}$.
However, essentially all output bits are conveniently available if
$\max_{\sigma\in\Sigma,j\in[m]}|\Asigj| \leq 2^k$, which is captured by the following
lemma.

\begin{lemma}\label{lem:preimage}
  Given $k\in\NN$ and $S'\subseteq S$, $\Sigma'\subseteq\Sigma$, if
  $\max_{\sigma\in\Sigma,j\in[m]}|\Asigj| \leq 2^k$ we have that
  \begin{equation}\label{eqn:preimage}
    \bigstarr_{s\in S', \sigma\in\Sigma'} o(s,\sigma)_j =
    \bigstarr_{s\in S', \sigma\in\Sigma'} e^{(s)}_{\Asigj}\,.
  \end{equation}
\end{lemma}
\begin{IEEEproof}
  As $2^k \geq \max_{\sigma\in\Sigma',j\in[m]}|\Asigj|$, by \mbox{Definition~\ref{def:universal}}
  every $\Asigj$ is encoded in $e^{(s)}$, i.e., every entry
  $e^{(s)}_{\Asigj}$ exists.

  Next, we distinguish three cases for every evaluation of the l.h.s.\ of~\eqref{eqn:preimage}: $1$, $0$, and $\mfu$.
  In the first case, $\bigstarr_{s\in S', \sigma\in\Sigma} o(s,\sigma)_j = 1$,
  we apply Definition~\ref{def:preimage} and Definition~\ref{def:universal} to show the
  claim, as follows.
  \begin{align*}
    & \bigstarr_{s\in S', \sigma\in\Sigma'} o(s,\sigma)_j = 1 \\
    \Leftrightarrow & \,\forall s\in S',\sigma\in\Sigma':\,o(s,\sigma)_j = 1 \\
    \Leftrightarrow & \,\forall s\in S',\sigma\in\Sigma':\,s\in\Asigj \tag{\small Def~\ref{def:preimage}}\\
    \Leftrightarrow & \,\forall s\in S',\sigma\in\Sigma':\,e^{(s)}_{\Asigj} = 1 \tag{\small Def~\ref{def:universal}}\\
    \Leftrightarrow & \bigstarr_{s\in S', \sigma\in\Sigma'} e^{(s)}_{\Asigj} = 1
  \end{align*}
  The second case, $\bigstarr_{s\in S', \sigma\in\Sigma} o(s,\sigma)_j = 0$, is
  treated analogously.
  \begin{align*}
    & \bigstarr_{s\in S', \sigma\in\Sigma'} o(s,\sigma)_j = 0 \\
    \Leftrightarrow & \,\forall s\in S',\sigma\in\Sigma':\,o(s,\sigma)_j = 0 \\
    \Leftrightarrow & \,\forall s\in S',\sigma\in\Sigma':\,s\notin\Asigj \tag{\small Def~\ref{def:preimage}}\\
    \Leftrightarrow & \,\forall s\in S',\sigma\in\Sigma':\,e^{(s)}_{\Asigj} = 0 \tag{\small Def~\ref{def:universal}}\\
    \Leftrightarrow & \bigstarr_{s\in S', \sigma\in\Sigma'} e^{(s)}_{\Asigj} = 0
  \end{align*}
  In case the l.h.s.\ of~\eqref{eqn:preimage} evaluates to $\mfu$ the statement
  follows from the first two cases. As we showed equivalence in case $1$ and $0$,
  we deduce, that in case the l.h.s.\ of~\eqref{eqn:preimage} evaluates to $\mfu$,
  the r.h.s.\ of~\eqref{eqn:preimage} also evaluates to $\mfu$.
\end{IEEEproof}

If the size of the largest preimage ($\max_{\sigma\in\Sigma,j\in[m]}|\Asigj|$) exceeds
$2^k$ not all preimages have a corresponding entry in $e^{(s_{i-1})}$.
We need to compute whether $s_{i-1}$ is in the preimage. For this purpose, we
define a cover of $\Asigj$ with sets of size at most $2^k$.

\begin{definition}\label{def:cover}
  For an input symbol $\sigma\in\Sigma$, $j\in[m]$, and $k\in\NN$ we define $\Asigjk$, the cover of $\Asigj$ containing
  sets of cardinality smaller or equal to $2^k$:
  \begin{equation*}
    \Asigjk \coloneqq
    \begin{cases}
      \{ \Asigj \} & \mbox{if } |\Asigj| \leq 2^k\,,\\
      \{ A\subseteq\Asigj| |A| = 2^k\} & \mbox{else.}
    \end{cases}
  \end{equation*}
\end{definition}

If the state of the transducer is in one of the sets in $\Asigjk$, then it is
also in $\Asigj$. All sets in $\Asigjk$ have a corresponding entry in the state vector.
We can take the $\OR$ over all entries to see whether the transducer is in a state
of $\Asigj$ and hence whether it outputs $1$.

\begin{lemma}\label{lem:cover}
  For $k\in\NN$, $j\in[m]$, and $S'\subseteq S$, $\Sigma'\subseteq\Sigma$
  such that $2^k < \max_{\sigma\in\Sigma,j\in[m]}|\Asigj|$,
  if $|S'|\leq 2^k$ we have that
  \begin{equation}\label{eqn:cover}
    \bigstarr_{s\in S', \sigma\in\Sigma'} o(s,\sigma)_j =
    \bigstarr_{\sigma\in\Sigma'} \bigvee_{A\in\Asigjk} \bigstarr_{s\in S'} e^{(s)}_{A}
  \end{equation}
\end{lemma}
\begin{IEEEproof}
  Every $A\in \Asigjk$ has cardinality smaller or equal to $2^k$. Hence, there is
  an entry in vector $e^{(s)}$ corresponding to $A$, i.e., entry $e^{(s)}_A$ exists in the encoding.

  From its definition we observe that $\Asigjk$ is indeed a cover of $\Asigj$, i.e.,
  \begin{align}
    \bigcup_{A\in \Asigjk} A &= \Asigj\,.\label{eqn:unionA}
  \end{align}
  Moreover, Definition~\ref{def:cover} ensures that every subset of $\Asigj$ of
  size at most $2^k$ is contained in at least one set from $\Asigjk$.
  On the other hand, each $A\in \Asigjk$ is a subset of $\Asigj$.
  Hence, the assumption that $|S'|\leq 2^k$ implies that
  \begin{align}
    S'\subseteq \Asigj &\Leftrightarrow \exists A\in \Asigjk:\,S'\subseteq
    A\,.\label{eqn:splitA}
  \end{align}

  We make a case distinction on every possible evaluation of the l.h.s.\ of~\eqref{eqn:cover} and show
  the equality for each case.
The first case is $\bigstarr_{s\in S', \sigma\in\Sigma'} e^{(s)}_{\Asigj} = 1$.
  We get that
  \begin{align*}
    & \bigstarr_{s\in S', \sigma\in\Sigma'} o(s,\sigma)_j = 1 \\
    \Leftrightarrow & \,\forall s\in S',\sigma\in\Sigma':\,o(s,\sigma)_j = 1 \\
    \Leftrightarrow & \,\forall s\in S',\sigma\in\Sigma':\,s\in\Asigj \\
    \Leftrightarrow & \,\forall \sigma\in\Sigma':\,S'\subseteq\Asigj\\
    \Leftrightarrow & \,\forall \sigma\in\Sigma':\,\exists A\in\Asigj:\,S'\subseteq A \tag{\small by \eqref{eqn:splitA}}\\
    \Leftrightarrow & \,\forall \sigma\in\Sigma':\,\exists A\in\Asigj:\, \forall s\in S':\, s \in A \\
    \Leftrightarrow & \,\forall \sigma\in\Sigma':\,\exists A\in\Asigj:\, \forall s\in S':\, e^{(s)}_{A} = 1 \\
    \Leftrightarrow & \,\forall \sigma\in\Sigma':\,\exists A\in\Asigj:\, \bigstarr_{s\in S'} e^{(s)}_{A} = 1 \\
    \Leftrightarrow & \,\forall \sigma\in\Sigma':\,\bigvee_{A\in\Asigjk} \bigstarr_{s\in S'} e^{(s)}_{A} = 1 \\
    \Leftrightarrow & \,\bigstarr_{\sigma\in\Sigma'} \bigvee_{A\in\Asigjk} \bigstarr_{s\in S'} e^{(s)}_{A} = 1.
  \end{align*}
  The second case, $\bigstarr_{s\in S', \sigma\in\Sigma'} e^{(s)}_{\Asigj} = 0$,
  is treated similarly, where now $S'\cap \Asigj = \emptyset$ for each $\sigma \in \Sigma'$.
  \begin{align*}
    & \bigstarr_{s\in S', \sigma\in\Sigma} o(s,\sigma)_j = 0 \\*
    \Leftrightarrow & \,\forall s\in S',\sigma\in\Sigma':\,o(s,\sigma)_j = 0 \\
    \Leftrightarrow & \,\forall s\in S',\sigma\in\Sigma':\,s\notin\Asigj \\
    \Leftrightarrow & \,\forall s\in S',\sigma\in\Sigma',A\in\Asigjk:\, s\notin A \\
    \Leftrightarrow & \,\forall s\in S',\sigma\in\Sigma',A\in\Asigjk:\, e^{(s)}_{A} = 0 \\
    \Leftrightarrow & \,\forall \sigma\in\Sigma',A\in\Asigjk:\, \bigstarr_{s\in S'} e^{(s)}_{A} = 0 \\
    \Leftrightarrow & \,\forall \sigma\in\Sigma':\, \bigvee_{A\in\Asigjk} \bigstarr_{s\in S'} e^{(s)}_{A} = 0 \\
    \Leftrightarrow & \bigstarr_{\sigma\in\Sigma'} \bigvee_{A\in\Asigjk} \bigstarr_{s\in S'} e^{(s)}_{A} = 0
  \end{align*}

  In the final case, i.e., that the l.h.s.\ of~\eqref{eqn:cover} evaluates to
  $\mfu$, equality follows from the equivalence established in the previous two
  cases.
\end{IEEEproof}
\subsubsection*{Main theorem}
We are left with the task of showing that indeed the obtained circuit is correct,
i.e., prove Theorem~\ref{thm:trans}.
The correctness of the construction is proven foremost by the application of
Corollary~\ref{cor:extension_commutes} and Corollary~\ref{cor:faithful}.
The proof is mainly concerned with establishing the size
and depth bound for the obtained circuit. Without going into detail, for constant $|S|$ the
reader should be convinced that the depth of the circuit is logarithmic in $n$
because all operations except for Step~2 can be computed in parallel, while
Step~2 exploits the associativity of matrix multiplication to obtain a circuit
of depth logarithmic in $n$. The size of the circuit is linear in $n$, as Step~1,
Step~3 and Step~4 each use a constant number of operations for each input symbol and Step~2
can be performed asymptotically optimally with a linear number of operations.
\begin{figure}\label{fig:circuit}
% inside_import 
% before   
% ignored 
% args [width=\linewidth]
% full_filename stage1-3
% after 
  \includegraphics[width=\linewidth]{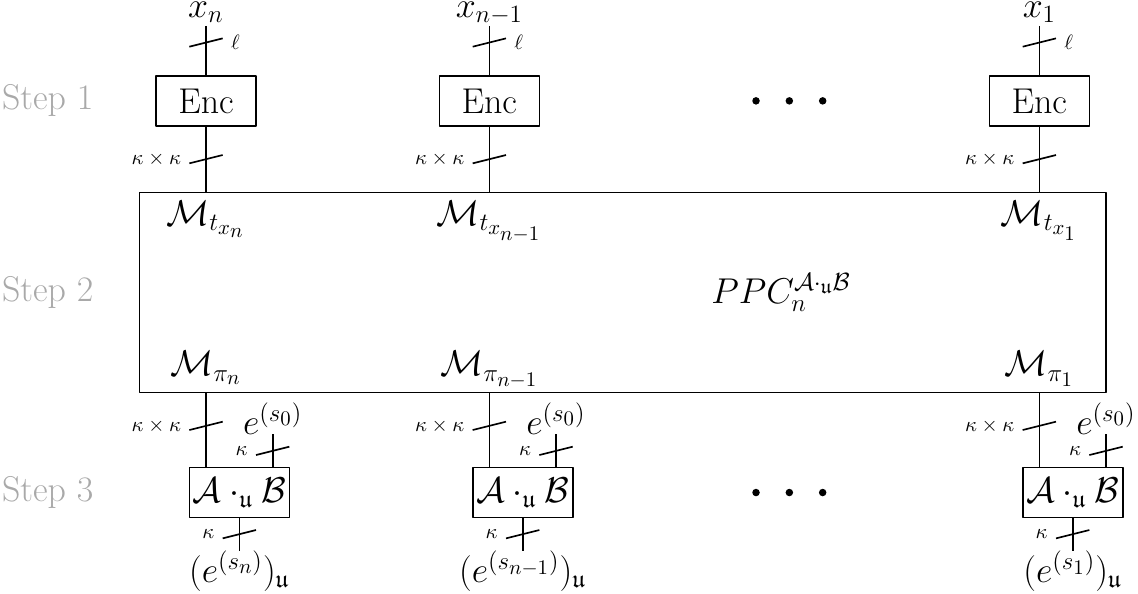}
  \caption{Steps $1$ to $3$ of the circuit implementing the transcription function.
  Enc denotes the computation of the universal encoding by the generic construction of \cite{jukna2021notes}.
  Hazard-free Boolean matrix multiplication is denoted by $\mathcal{A}\cdot_\mfu\mathcal{B}$.
}
\end{figure}
\mainres
\begin{IEEEproof}
We show that there is a circuit computing the hazard-free extension of the transcription
  function
  $\left(\tau_{T,n}\right)_{\mfu}(x)$ for every $x\in\Sigma^n$, where we replace at most $k$
  many bits with $\mfu$'s. We follow the steps of the PPC framework presented in Section~\ref{subsec:ppc}
  to compute output $((\tau_{T,n})_{\mfu}(x))_i=o_{\mfu}(s_{i-1},x_i)$ at each position
  $i\in\{1\ldots n\}$.

  The block diagram of steps $1$ to $3$ of the circuit is depicted in Figure~\ref{fig:circuit}.
  The circuit mostly consists of hazard-free matrix multiplication
  blocks. In particular, the $n$-input parallel prefix circuit is an arrangement of
  hazard-free matrix multiplication blocks.

  Step~$1$ computes the universal encoding of the restricted
  transition function $t_{x_i}$ (a $\kappa\times\kappa$ matrix)
  from input $x_i$.
  Noting that the computation in Step~$1$ evaluates a function from $\IB^{\ell}$ to $\IB^{\kappa^2}$,
  we can directly apply Theorem~\ref{thm:jukna} for each output bit seperately.
  Hence, there is a hazard-free encoding circuit of size $\BO((2^\ell / \ell)\kappa^2)$
  and depth $\BO(\ell)$ implementing this step.
  Thus, the computation of the universal encoding of
  ${t_{x_i}}$ for each $i$ in parallel has size
  $\BO((2^\ell/\ell)\kappa^2n)$ and depth $\BO(\ell)$.

  Step~$2$ computes the encoding $\mathcal{M}_{\pi_i}$ of the composition
  $\pi_i=t_{x_i}\circ\ldots\circ t_{x_1}$.
  We define $E_j(x_j)=t(\cdot,x_j)$ for $j\in \{1,\ldots,i\}$ such that
  for $E(x)=\circ^i_{j=1}E_j(x_j)$ we have
  \begin{equation*}
    \M{\pi_i}_\mfu
    = \M{E(x)}_\mfu
    = \M{E(\cdot)}_\mfu(x)\,.
  \end{equation*}
  Application of Corollary~\ref{cor:extension_commutes} then yields
  \begin{align*}\label{eqn:prefix}
   & \M{E(\cdot)}_{\mfu}(x)\\&=\M{E_j(\cdot)}_{\mfu}(x_j)\cdot_{\mfu}
    \M{E_{j-1}(\cdot)}_{\mfu}(x_{j-1})\\&\quad\quad\cdot_{\mfu}\ldots
    \cdot_{\mfu}\M{E_1(\cdot)}_{\mfu}(x_1)\,.\tag{9}
  \end{align*}
  By Corollary~\ref{col:ppc}, there is an efficient circuit computing the r.h.s.\ of
  \eqref{eqn:prefix} for each $i$.
  The Corollary gives also size $\BO(\kappa^3n)$ and depth $\BO(\log\kappa\log n)$ for Step~$2$.

  Step~$3$ computes the column unit vector corresponding to the $i$-th state, i.e.,
  evaluation of the composition $\pi_i$ on the initial state $s_0$.
  By Corollary~\ref{cor:faithful}, we can compute
$\left(e^{(s_i)}\right)_\mfu$
  by matrix-vector multiplication of $\M{\pi_i}_\mfu$ and
  $\left(e^{(s_0)}\right)_\mfu=e^{(s_0)}$. Hence, by Corollary~\ref{cor:monotone}
  there is a $k$-bit hazard-free circuit that computes $s_i$.
  Evaluation of $\pi_i$ in parallel for each $i$ yields size $\BO(\kappa^2n)$ and depth $\BO(\log\kappa)$ for Step~$3$.

  \begin{figure}\label{fig:outputstage}
% inside_import 
% before     
% ignored 
% args [width=\linewidth]
% full_filename stage4
% after 
    \includegraphics[width=\linewidth]{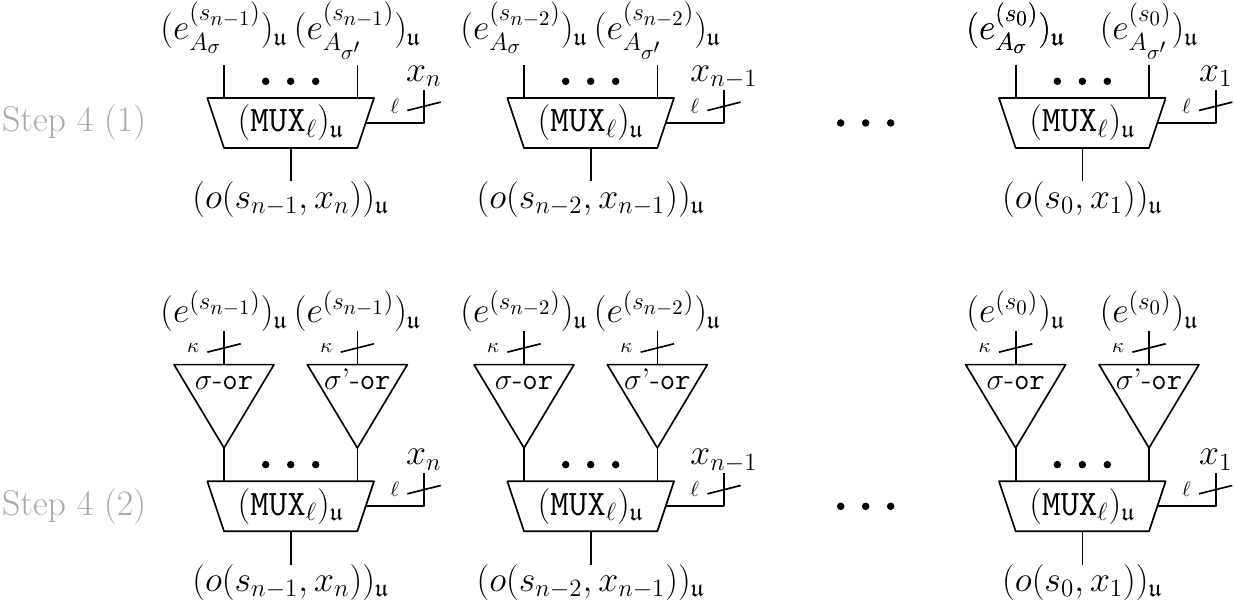}
    \caption{Step~4 of the circuit implementing the transcription function, assuming $m=1$, for
    the case that every preimage is encoded in the state vector (1) and the case that there is
    at least one preimage not encoded in the state vector (2).
    To increase readability we do not enumerate all elements of $\Sigma$ but use
    $\sigma$, $\sigma'$, and dots to denote that the step is repeated for every element of $\Sigma$.}
  \end{figure}

  Finally, Step~$4$ computes the $i$-th output $o(s_{i-1}, x_i)$ for each
  $i\in\{1\ldots n\}$.
W.l.o.g., assume that the width of an output symbol is $1$, i.e., $m=1$ and
  hence $\Lambda\subseteq\IB$ (otherwise repeat
  the computation for each output bit separately).
  Viewing the computation in Step~$4$ as a function from $\IB^{\kappa + \ell}$ to $\IB$, we could
  apply Theorem~\ref{thm:jukna} to obtain a circuit of size $\BO(2^{\kappa+\ell}/(\kappa+\ell))$
  and depth $\BO(\kappa + \ell)$. We now show how to obtain better results, bounding the
  size by $\BO(2^\ell \kappa)$ with a circuit of depth $\BO(\log\kappa + \ell)$.

  Recall Definition~\ref{def:preimage}.  As $j=0$ we omit $j$ from the
  notation and write $\Asig$ and $\Asigk$ instead of $\Asigj$ and $\Asigjk$.
  First, consider the case that
  $\max_{\sigma\in\Sigma} |\Asig|\leq 2^k$.
  A depiction of this case is given in Figure~\ref{fig:outputstage} (1).
  By directly using the outputs of the gates computing the respective bits of
  (the universal encoding of) the state vector, we readily obtain $o(s_{i-1},
  \sigma)$ for each $\sigma\in\Sigma$.
  Thus, we are left with the
  task of choosing the output corresponding to input $x_i$. We can do so by using
  a $\MUX$. By Corollary~\ref{col:cmux}, the implementation
  $\MUX(e^{(s_{i-1})}_{\Asig},\ldots,e^{(s_{i-1})}_{A_{\sigma'}},x_i)$
  has size $\BO(2^\ell)$ and depth $\BO(\ell)$, where $\sigma$, $\sigma'$ are representatives
  for every input sympol in $\Sigma$. If $m>1$, the multiplexer is copied $m$
  times and wired accordingly. Step~4 can be performed in parallel for each $i$,
  hence the resulting size and depth bounds are $\BO(2^\ell m n)$ and
  $\BO(\ell)$, respectively.

  The other case is that $\max_{\sigma\in\Sigma} |\Asig|> 2^k$.
  A depiction of the corresponding circuit is given in
  Figure~\ref{fig:outputstage} (2), where $\operatorname{\sigma-\OR}$ denotes
  the $\OR$-tree over all $e^{(s_i)}_{A}$ for $A\in\Asigk$.
  Here, we compute the output $o(s_{i-1}, \sigma)$ as the $\OR$
  over all entries corresponding to $\Asigk$ in the state vector.
  Correctness readily follows from Lemma~\ref{lem:cover}.
  To bound the size and depth of the resulting circuit, observe that the cardinality
  of $\Asigk$ is bounded by $\binom{|S|}{2^k}$, hence each $\OR$-tree has size
  $\BO(\binom{|S|}{2^k})$ and depth $\BO(\log\binom{|S|}{2^k})$.
  As in the previous case, the $i$-th output is selected by a multiplexer.
  Hence, applying Corollary~\ref{col:cmux}, we get that in this case Step~4 requires size
  $\BO(\binom{|S|}{2^k} 2^\ell m n)$ and depth $\BO(\log\binom{|S|}{2^k} +
  \ell)$.

Furthermore, there is a natural cap on $m$. Similar to the previous paragraph,
  consider each bit of the output separately, i.e., for each position of an
  output symbol consider the function $o\colon S\times\Sigma\rightarrow\IB$.
  For inputs from $S$ and $\Sigma$ there are $2^{|S|\cdot|\Sigma|}$ different one bit functions.
  Hence, we can enumerate all possible output functions. If
  $m>2^{|S|\cdot|\Sigma|}$, we simply compute and reuse as often as needed the
  output of each of these possible functions, rather than replicating
  computations for identical output bits.

  Finally, we can derive the asymptotic size and depth of the presented
  $k$-bit hazard-free implementation of the transcription function $\tau_{T,n}$.
  We distinguish two cases depending on the size of the largest preimage of $o$.
  In both cases $m$ is capped at $2^{|S|\cdot|\Sigma|}$ as discussed above.
  Assume $\max_{\sigma\in\Sigma} |\Asig|\leq 2^k$, then the presented circuit has
  \begin{align*}
    &\text{size: }
    \BO\left((\kappa^3 + (2^{\ell}/\ell)\kappa^2 + 2^\ell m) n\right)
    \text{, and}\\
    &\text{depth: }
    \BO\left(\log\kappa\log n + \ell\right)
    \,.
  \end{align*}
  In case $\max_{\sigma\in\Sigma} |\Asig| > 2^k$ (and hence $|S|>2^k$) we bound $\binom{|S|}{2^k}$ by $\kappa$,
  such that the circuit has size
  \begin{equation*}
\BO\left((\kappa^3 + (2^{\ell}/\ell)\kappa^2 + 2^\ell \kappa m) n\right)\,,
  \end{equation*}
  and depth
  \begin{equation*}
\BO\left(\log\kappa\log n + \ell\right)\,.\IEEEQEDhereeqn
  \end{equation*}
\end{IEEEproof}

We remark that the main result holds also for $\Sigma\subseteq \IB^{\ell}$ when
choosing an (arbitrary) extension for the transition function to domain $S\times\IB^\ell$. However, the choice of how to
extend the transition function $t$ matters for
the behavior of the hazard-free state machine.
The decision on how to treat non-input symbols is important because it is possible that an
unstable input resolves to such a non-input symbol. If this happens, the choice of where
$t$ maps such resolutions to affects the value the hazard-free extension takes. A
``bad'' extension may result in unstable output without need, decreasing the
utility of the constructed circuit.
The task of finding a useful extension is nontrivial and depends on the application.
A detailed discussion at the hand of an example is given in Section~\ref{sec:inp_enc}.

 % end input ./example.tex
 % start input ./appendix_a.tex
\section{Extension of the input encoding}\label{sec:inp_enc}

The main result, Theorem~\ref{thm:trans}, assumes that each binary string of
length $\ell$ is part of the input alphabet, i.e., $\Sigma=\IB^\ell$.
In this section, we discuss the case where $\Sigma$ is only a strict subset of $\IB^\ell$.
We choose extensions of the transition function $t\colon\Sigma\times S$
and ouput function $o\colon\Sigma\times S$ to domain $\IB^\ell\times S$.
The choice of extension is arbitrary, for any transition function and
output function defined on $\IB^\ell$ the
construction described in the proof of the main theorem computes the
hazard-free extension of the transcription function.
An arbitrary extension may lead to undesirable results, however.
We discuss this issue at hand of an example.
Note that both, $t$ and $o$ have to be extended and the same problems may arise
in both functions.
We construct an example, that provides intuition on the problem, by extending the
output function $o$. The transition function $t$ is not discussed in further detail.

\subsubsection*{Example}
Consider a simple finite-state transducer with a single state ($S=\{0\}$) with
single bit outputs and output alphabet $\Lambda=\IB$.
Input symbols have three bits ($\ell=3$) and the output function computes the
$\AND$ over the three input bits\footnote{The $\AND$ over the three bits is defined by two sequential
$\AND$ gates (as given in Table~\ref{tab:gates}).}. See Table~\ref{tab:apx_out}
for the output table.
For the sake of the example, assume that input $000$ never occurs, such that the input alphabet is
defined by $\Sigma=\IB^3 \setminus \{000\}$.
Even though the example is artificial, in the sense that it could be solved by
simpler circuits, it is worth the discussion as it provides intuition to the problem.

Unknown inputs (which we model by $\mfu$) may arise from transitioning input signals.
Due to the analog behavior of
signal transitions, it might be the case that a bit is neither digital $0$ nor $1$
for some time during the transition. The signal can float near the threshold for
$0$ or $1$ and the circuit designer can not be sure whether the electronics
interpret the signal as digital $0$ or $1$.

\subsubsection*{Extension of the output function}
For an application of the framework, we have to define an extension to $o$ on input $000$. As there is only
a single output bit, we are left with exactly two choices:
\begin{align}
  o(0,000) &= 0\tag{1}\,\text{, or}\\
  o(0,000) &= 1\tag{2}\,.
\end{align}

Assume the input switches from $001$ to $010$, such that the input transitions via
$0\mfu\mfu$ for some time.
We choose extension $(2)$, then the framework
computes the hazard-free extension of the output function
\begin{align*}
  &o_\mfu(0,0\mfu\mfu)\\
  =&\bigstarr \{ o(0,000),o(0,001),o(0,010),o(0,011)\}\\
  =&\bigstarr\{1,0\}\\
  =&\:\mfu\,.
\end{align*}
In the sense of hazard-freeness, this is perfectly fine, because we compute the
most precise output for all specified inputs and the extension. However,
input $001$ and $010$ both produce output $0$. Hence, the output has no transition
when switching from $001$ to $010$. But, due to the chosen extension, the output is $\mfu$
during the transition.

\begin{table}
  \centering
  \caption{Output function of the example transducer and extensions $(1)$ and $(2)$.}
  \label{tab:apx_out}
  \begin{tabular}{c|c|c|c}
    $\sigma$ & $o(\sigma,0)$ & $(1)$ & $(2)$ \\ \hline
    $000$ & - & $0$ & $1$ \\
    $001$ & $0$ & $0$ & $0$ \\
    $010$ & $0$ & $0$ & $0$ \\
    $100$ & $0$ & $0$ & $0$ \\
    $011$ & $0$ & $0$ & $0$ \\
    $101$ & $0$ & $0$ & $0$ \\
    $110$ & $0$ & $0$ & $0$ \\
    $111$ & $1$ & $1$ & $1$
  \end{tabular}
\end{table}

Contrary to that, if we choose extension $(1)$ we obtain
\begin{align*}
  &o_\mfu(0,0\mfu\mfu)\\
  =&\bigstarr \{ o(0,000),o(0,001),o(0,010),o(0,011)\}\\
  =&\bigstarr\{0\}\\
  =&0\,.
\end{align*}
Hence, the output signal remains stable during the transition of the inputs.
Also one can verify that for each input transition, the output remains unless
there is a transition in the output (for any input transition to or from $111$).
Hence, the advisable choice of the extension for the specific application is
extension $(1)$.

To conclude, the example shows that finding an extension of the transition function,
in case $\Sigma$ is a strict subset of $\IB^\ell$, is not a straightforward task.
We refrain from defining a notion of the optimum choice and
generating the according extension, as this is beyond the scope of this work.

 % end input ./appendix_a.tex
 % start input ./appendix_b.tex
\section{Bound on $k$}\label{sec:example_sec}

We present an extension of our earlier example, the
shift transducer. Here, we switch from output $0$ to output $1$ only
if we see two consecutive $1$'s in the input. Vice versa, we switch from $1$'s
to $0$'s only we see two consecutive $0$ inputs. The transducer has three states
$S=\{0,1,2\}$. It is depicted in Figure~\ref{fig:ex_sec_trans}.

What we seek to demonstrate is the impact of bounding $k$.
Accordingly, we choose $k=1$, such that $2^k=2<3=|S|$. Thus, the encoding takes into account
$\mathcal{P}_2(S)=\mathcal{P}(S)\setminus S$, i.e., all proper subsets of $S$.
For restricted transition functions $t_0$, $t_1$, and $t_{\mfu}$, the encoding
yields the following matrices.

\begin{align*}
  \begin{blockarray}{cccccccc}
   \mathcal{M}_{t_0} =\quad & \emptyset & \{0\} & \{1\} & \{2\} & \{0,1\} & \{1,2\} & \{0,2\}\\
    \begin{block}{c(ccccccc)}
      \emptyset & 1 & 0 & 0 & 0 & 0 & 0 & 0 \\
      \{0\}   & 1 & 1 & 1 & 0 & 1 & 0 & 0 \\
      \{1\}   & 1 & 0 & 0 & 1 & 0 & 0 & 0 \\
      \{2\}   & 1 & 0 & 0 & 0 & 0 & 0 & 0 \\
      \{0,1\} & 1 & 1 & 1 & 1 & 1 & 1 & 1 \\
      \{1,2\} & 1 & 0 & 0 & 1 & 0 & 0 & 0 \\
      \{0,2\} & 1 & 1 & 1 & 0 & 1 & 0 & 0 \\
    \end{block}
  \end{blockarray} \\
\begin{blockarray}{cccccccc}
    \mathcal{M}_{t_1}=\quad& \emptyset & \{0\} & \{1\} & \{2\} & \{0,1\} & \{1,2\} & \{0,2\}\\
    \begin{block}{c(ccccccc)}
      \emptyset & 1 & 0 & 0 & 0 & 0 & 0 & 0 \\
      \{0\}   & 1 & 0 & 0 & 0 & 0 & 0 & 0 \\
      \{1\}   & 1 & 1 & 0 & 0 & 0 & 0 & 0 \\
      \{2\}   & 1 & 0 & 1 & 1 & 0 & 1 & 0 \\
      \{0,1\} & 1 & 1 & 0 & 0 & 0 & 0 & 0 \\
      \{1,2\} & 1 & 1 & 1 & 1 & 1 & 1 & 1 \\
      \{0,2\} & 1 & 0 & 1 & 1 & 0 & 1 & 0 \\
    \end{block}
  \end{blockarray} \\
\begin{blockarray}{cccccccc}
    \mathcal{M}_{t_{\mfu}}=\quad& \emptyset & \{0\} & \{1\} & \{2\} & \{0,1\} & \{1,2\} & \{0,2\}\\
    \begin{block}{c(ccccccc)}
      \emptyset & 1 & 0 & 0 & 0 & 0 & 0 & 0 \\
      \{0\}   & 1 & \mfu & \mfu & 0 & \mfu & 0 & 0 \\
      \{1\}   & 1 & \mfu & 0 & \mfu & 0 & 0 & 0 \\
      \{2\}   & 1 & 0 & \mfu & \mfu & 0 & \mfu & 0 \\
      \{0,1\} & 1 & 1 & \mfu & \mfu & \mfu & \mfu & \mfu \\
      \{1,2\} & 1 & \mfu & \mfu & 1 & \mfu & \mfu & \mfu \\
      \{0,2\} & 1 & \mfu & 1 & \mfu & \mfu & \mfu & 0 \\
    \end{block}
  \end{blockarray}
\end{align*}

In matrix $\mathcal{M}_{t_{\mfu}}$, we see again the same behavior as in the
naive encoding of the shift transducer. Columns corresponding to $\{0,1\}$,
$\{1,2\}$, and $\{0,2\}$ contain only $0$ and $\mfu$ entries. These columns
hence cannot prevent the resolution from containing the all $0$'s vector, which
cannot be mapped to a single state in a way that guarantees correct output.

\begin{figure}
  \centering
  \includegraphics{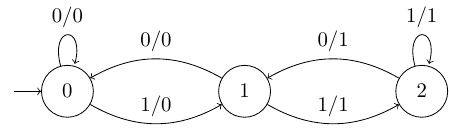}
  \caption{Extension of the first example transducer, that shifts to outputting
  ones only after reception of two consecutive ones. Vice versa the transducer
  shifts to outputting zeros upon reception of two consecutive zeros.}
  \label{fig:ex_sec_trans}
\end{figure}

However, columns corresponding to $\{0\}$, $\{1\}$, and $\{2\}$ can be mapped to
states without uncertainty, presuming the next input symbol allows it.
Taking a look at Figure~\ref{fig:ex_sec_trans}, we can see that the transducer
goes back to state $0$ after reception of two consecutive $0$ inputs, regardless
of the previous state. For the sake of this example, consider the input
sequence $\mfu 0 0$. Hence, we need to calculate the matrix corresponding to
$t_0\circ t_0\circ t_{\mfu}$:

\begin{align*}
  \mathcal{M}_{t_0\circ t_0\circ t_{\mfu}} = \mathcal{M}_{t_0}\cdot\mathcal{M}_{t_0}\cdot\mathcal{M}_{t_{\mfu}}
=\begin{pmatrix}
    1 & 0 & 0 & 0 & 0 & 0 & 0 \\
    1 & 1 & 1 & 1 & \mfu & \mfu & \mfu \\
    1 & 0 & 0 & 0 & 0 & 0 & 0 \\
    1 & 0 & 0 & 0 & 0 & 0 & 0 \\
    1 & 1 & 1 & 1 & \mfu & \mfu & \mfu \\
    1 & 0 & 0 & 0 & 0 & 0 & 0 \\
    1 & 1 & 1 & 1 & \mfu & \mfu & \mfu \\
  \end{pmatrix}.
\end{align*}

Here it becomes apparent why we need the assumption $|A|\cdot|J|\leq 2^k$ in
Theorem~\ref{thm:trans}.
In matrix $\mathcal{M}_{t_0\circ t_0\circ t_{\mfu}}$ only column unit vectors
corresponding to singletons can be mapped to stable states. Column vectors corresponding
to sets $\{0,1\}$, $\{1,2\}$ and $\{0,2\}$ contain only $0$'s and $\mfu's$, no $1$'s.
Hence, they cannot be ``properly'' mapped to any state.
We apply Theorem~\ref{thm:trans} to the last step, i.e., the multiplication of
$\mathcal{M}_{t_0}$ and $\mathcal{M}_{t_0\circ t_{\mfu}}$. Here, set $J$ is mapped to
$\res(0\mfu)$ yielding $|J|=2$. As $k=1$, condition $|A|\cdot|J|\leq 2^k$ is only satisfied
for $|A|=1$ ($A$ cannot be the empty set). Thus, the condition is only satisfied
for singleton sets, which is what we observed from the example.

 % end input ./appendix_b.tex
 % start input ./conclusion.tex
\section{Conclusion}\label{sec:conclusion}
Any generic construction for hazard-free circuits incurs an exponential blow-up
in circuit complexity~\cite{ikenmeyer18complexity}.
In this work, we present a generic construction for transducers that is
asymptotically optimal in the size of the inputs $n$, but yields an exponential
overhead in the size of the transducer. By Theorem~\ref{thm:trans}, for a
transducer with $|S|$ states we obtain a circuit of size
$\BO(\kappa^3n + 2^{\ell}\kappa^2n + 2^{\ell}\kappa\lambda n)$ and depth
$\BO(\log\kappa\log n + \ell)$, where $\ell$ is the maximum number of bits of an
input symbol, $k$ is the maximum number of uncertain bits in the input, and
$\kappa^2$ is the size of the universal function encoding.
The universal encoding of functions is a key ingredient in the construction
of the circuit.
Together with Theorem~\ref{thm:extension_commutes}, we show
the correctness of the construction. The Theorem states that,
for the chosen encoding, the superposition of function composition is the
matrix product of the superpositions of both functions.
We specify our findings in Corollary~\ref{cor:general} and Corollary~\ref{cor:smallk}.

This opens the stage for a $k$-bit hazard-free implementation of addition.
We remark that to make hazard-freeness meaningful in the context of addition, one has to limit the
uncertainty of the input concerning the encoded sum and choose appropriate
encodings. The relevant definitions exceed the scope of this submission.
In a preliminary publication~\cite{bund2018small}, we show that it is possible to apply this construction to obtain efficient
circuits for addition that avoid certain hazards. This
demonstrates that the construction we present here is of interest and, in our
view, has surprising consequences.

 % end input ./conclusion.tex
 
\ifCLASSOPTIONcaptionsoff
  \newpage
\fi

\bibliographystyle{IEEEtran}
\bibliography{IEEEabrv, comp}

\end{document}